\documentclass[sigconf]{acmart}

\usepackage{graphicx}
\usepackage{subcaption}
\usepackage{tikz}
\usepackage{calc}

\usepackage{array}
\usepackage{tabularx}

\usepackage{url}

\usepackage{balance}

\usepackage{listings}

\usepackage{xspace}

\usepackage{cleveref}

\usepackage{paralist}

\usepackage{amsmath}

\usepackage{amssymb}
\usepackage{amsthm,mathtools}

\usepackage{comment}
\excludecomment{omit}

\newcommand{\cloudone}[0]{Amazon\xspace}
\newcommand{\sysname}[0]{\textsc{RNG}\xspace}

\newcommand{\change}[1]{#1}
\newcommand{\OMIT}[1]{}
\newcommand{\newtext}[1]{#1}

\newcommand{\parab}[1]{\vspace{2pt}\noindent\textbf{#1}.\xspace}

\makeatletter
\renewcommand{\sectionautorefname}{\S\@gobble}
\renewcommand{\subsectionautorefname}{\S\@gobble}
\renewcommand{\subsubsectionautorefname}{\S\@gobble}
\renewcommand{\appendixautorefname}{\S\@gobble}
\makeatother

\renewcommand\footnotetextcopyrightpermission[1]{}
\setcopyright{none}

\makeatletter
\def\@authorfont{\large}
\def\@affiliationfont{\large\sffamily}
\def\@typeset@author@bx{\bgroup\hsize=\author@bx@wd\def\and{\par}%
  \global\setbox\author@bx=\vtop{\if@ACM@sigchiamode\else\centering\fi
    \@authorfont\@currentauthors\par\vspace{4pt}\@affiliationfont
    \@currentaffiliation}\egroup
  \box\author@bx\hspace{\author@bx@sep}%
  \gdef\@currentauthors{}%
  \gdef\@currentaffiliation{}}
\makeatother

\acmDOI{}

\acmISBN{}

\acmPrice{}

\title{RNG: Flat Datacenter Networks at Scale}

\author{\Large
  {Giacomo Bernardi \qquad 
  Ratul Mahajan}\textsuperscript{\dag} \qquad C.~Seshadhri\textsuperscript{\ddag}\vspace{-6pt}}
  \author{\Large
  \setlength{\tabcolsep}{14pt}%
  \begin{tabular}{@{}cccc@{}}
    Enrico Carlesso & Chinchu Merine Joseph & Saurabh Kumar    & Pavan Manikonda\\[4pt]
    Luiza Popa      & Randy Ram             & Steven Robinson  & Elizabeth Tennent\\
  \end{tabular}}
\affiliation{
  \Large\institution{Amazon Web Services}}

\renewcommand{\shortauthors}{}

\begin{document}

\newtheorem{claim}[theorem]{Claim}
\newtheorem{fact}[theorem]{Fact}
\newtheorem{assumption}[theorem]{Assumption}
\newtheorem{observation}[theorem]{Observation}
\newtheorem{remark}[theorem]{Remark}
\newtheorem{model}{Model}[section]

\newcommand{\ignore}[1]{}

\newcommand{\cA}{{\cal A}}
\newcommand{\cB}{\mathcal{B}}
\newcommand{\cC}{{\cal C}}
\newcommand{\cD}{\mathcal{D}}
\newcommand{\cE}{\mathcal{E}}
\newcommand{\cF}{\mathcal{F}}
\newcommand{\cG}{\mathcal{G}}
\newcommand{\cH}{{\cal H}}
\newcommand{\cI}{{\cal I}}
\newcommand{\cJ}{{\cal J}}
\newcommand{\cL}{{\cal L}}
\newcommand{\cM}{{\cal M}}
\newcommand{\cP}{\mathcal{P}}
\newcommand{\cQ}{\mathcal{Q}}
\newcommand{\cR}{{\cal R}}
\newcommand{\cS}{\mathcal{S}}
\newcommand{\cT}{{\cal T}}
\newcommand{\cU}{{\cal U}}
\newcommand{\cV}{{\cal V}}
\newcommand{\cW}{{\cal W}}
\newcommand{\cX}{{\cal X}}

\newcommand{\HH}{\mathbb H}
\newcommand{\R}{\mathbb R}
\newcommand{\N}{\mathbb N}
\newcommand{\F}{\mathbb F}
\newcommand{\Z}{{\mathbb Z}}
\newcommand{\eps}{\varepsilon}
\newcommand{\lam}{\lambda}
\newcommand{\var}{\mathrm{var}}
\newcommand{\sgn}{\mathrm{sgn}}
\newcommand{\poly}{\mathrm{poly}}
\newcommand{\polylog}{\mathrm{polylog}}
\newcommand{\littlesum}{\mathop{{\textstyle \sum}}}
\newcommand{\half}{{\textstyle \frac12}}
\newcommand{\la}{\langle}
\newcommand{\ra}{\rangle}
\newcommand{\wh}{\widehat}
\newcommand{\wt}{\widetilde}
\newcommand{\calE}{{\cal E}}
\newcommand{\calL}{{\cal L}}
\newcommand{\calF}{{\cal F}}
\newcommand{\calW}{{\cal W}}
\newcommand{\calH}{{\cal H}}
\newcommand{\calN}{{\cal N}}
\newcommand{\calO}{{\cal O}}
\newcommand{\calP}{{\cal P}}
\newcommand{\calV}{{\cal V}}
\newcommand{\calS}{{\cal S}}
\newcommand{\calT}{{\cal T}}
\newcommand{\calD}{{\cal D}}
\newcommand{\calC}{{\cal C}}
\newcommand{\calX}{{\cal X}}
\newcommand{\calY}{{\cal Y}}
\newcommand{\calZ}{{\cal Z}}
\newcommand{\calA}{{\cal A}}
\newcommand{\calB}{{\cal B}}
\newcommand{\calG}{{\cal G}}
\newcommand{\calI}{{\cal I}}
\newcommand{\calJ}{{\cal J}}
\newcommand{\calR}{{\cal R}}
\newcommand{\calK}{{\cal K}}
\newcommand{\calU}{{\cal U}}
\newcommand{\barx}{\overline{x}}
\newcommand{\bary}{\overline{y}}

\newcommand{\bone}{{\bf 1}}
\newcommand{\ba}{\mathbf{a}}
\newcommand{\bb}{\mathbf{b}}
\newcommand{\bc}{\mathbf{c}}
\newcommand{\bd}{\mathbf{d}}
\newcommand{\be}{\mathbf{e}}
\newcommand{\bff}{\mathbf{f}}
\newcommand{\bg}{\mathbf{g}}
\newcommand{\bh}{\mathbf{h}}
\newcommand{\bi}{\mathbf{i}}
\newcommand{\bj}{\mathbf{j}}
\newcommand{\bk}{\mathbf{k}}
\newcommand{\bl}{\mathbf{l}}
\newcommand{\bmm}{\mathbf{m}}
\newcommand{\bn}{\mathbf{n}}
\newcommand{\bo}{\mathbf{o}}
\newcommand{\bp}{\mathbf{p}}
\newcommand{\bq}{\mathbf{q}}
\newcommand{\br}{\mathbf{r}}
\newcommand{\bs}{\mathbf{s}}
\newcommand{\bt}{\mathbf{t}}
\newcommand{\bu}{\mathbf{u}}
\newcommand{\bv}{\mathbf{v}}
\newcommand{\bw}{\mathbf{w}}
\newcommand{\bx}{\mathbf{x}}
\newcommand{\by}{\mathbf{y}}
\newcommand{\bz}{\mathbf{z}}

\newcommand{\bA}{\boldsymbol{A}}
\newcommand{\bD}{\boldsymbol{D}}
\newcommand{\bG}{\boldsymbol{G}}
\newcommand{\bH}{\boldsymbol{H}}

\newcommand{\bR}{\boldsymbol{R}}
\newcommand{\bS}{\boldsymbol{S}}
\newcommand{\bX}{\boldsymbol{X}}
\newcommand{\bY}{\boldsymbol{Y}}
\newcommand{\bZ}{\boldsymbol{Z}}

\newcommand{\NN}{\mathbb{N}}
\newcommand{\RR}{\mathbb{R}}
\newcommand{\mbone}{\mathbbm{1}}

\newcommand{\abs}[1]{\left\lvert #1 \right\rvert}
\newcommand{\norm}[1]{\left\lVert #1 \right\rVert}
\newcommand{\ceil}[1]{\lceil#1\rceil}
\newcommand{\Exp}{\EX}
\newcommand{\floor}[1]{\lfloor#1\rfloor}
\newcommand{\cei}[1]{\lceil#1\rceil}

\newcommand{\EX}{\mathbf{E}}
\newcommand{\prob}{{\rm Prob}}
\newcommand{\supp}{\mathrm{supp}}
\newcommand{\eqdef}{:=}

\newcommand{\gset}{Y}
\newcommand{\gcol}{{\cal Y}}

\newcommand{\otilde}{\widetilde{O}}

\newcommand{\Pois}{\mathrm{Pois}}

\newcommand{\nbr}{{{Nbr}}}
\newcommand{\wayp}{{{WP}}}
\newcommand{\ring}{{{Ring}}}
\newcommand{\inner}{{\iring}}
\newcommand{\out}{{\oring}}
\newcommand{\iring}{{{IR}}}
\newcommand{\oring}{{{OR}}}

\newcommand{\spg}[2]{\cS_{#1, #2}}
\newcommand{\point}[1]{\cP_{#1}}
\newcommand{\desc}[1]{D_{#1}}

\newcommand{\Sec}[1]{\hyperref[sec:#1]{\S\ref*{sec:#1}}}
\newcommand{\Eqn}[1]{\hyperref[eq:#1]{(\ref*{eq:#1})}}
\newcommand{\Fig}[1]{\hyperref[fig:#1]{Fig.\,\ref*{fig:#1}}}
\newcommand{\Tab}[1]{\hyperref[tab:#1]{Tab.\,\ref*{tab:#1}}}
\newcommand{\Thm}[1]{\hyperref[thm:#1]{Theorem\,\ref*{thm:#1}}}
\newcommand{\Fact}[1]{\hyperref[fact:#1]{Fact\,\ref*{fact:#1}}}
\newcommand{\Lem}[1]{\hyperref[lem:#1]{Lemma\,\ref*{lem:#1}}}
\newcommand{\Prop}[1]{\hyperref[prop:#1]{Prop.~\ref*{prop:#1}}}
\newcommand{\Cor}[1]{\hyperref[cor:#1]{Corollary~\ref*{cor:#1}}}
\newcommand{\Conj}[1]{\hyperref[conj:#1]{Conjecture~\ref*{conj:#1}}}
\newcommand{\Def}[1]{\hyperref[def:#1]{Definition~\ref*{def:#1}}}
\newcommand{\Alg}[1]{\hyperref[alg:#1]{Alg.~\ref*{alg:#1}}}
\newcommand{\Ex}[1]{\hyperref[ex:#1]{Ex.~\ref*{ex:#1}}}
\newcommand{\Clm}[1]{\hyperref[clm:#1]{Claim~\ref*{clm:#1}}}
\newcommand{\Step}[1]{\hyperref[step:#1]{Step~\ref*{step:#1}}}
\newcommand{\Mod}[1]{\hyperref[mod:#1]{Model~\ref*{mod:#1}}}

\newcommand{\Omgt}{\widetilde{\Omega}}

\begin{abstract}

We design and deploy 
in production
the first flat datacenter networks. 
Our design, called \sysname, is based on quasi-random graphs. 
While the cost and fault-tolerance benefits of such topologies have been long known, their practical realization has been hampered by a lack of scalable routing and cabling approaches. 
\sysname has a new distributed routing protocol that exploits the properties of random graphs to find a large number of edge disjoint paths between pairs of endpoints. 
It uses a novel passive optical device that internally shuffles cables, which 
makes its cabling complexity similar to that of fat trees. 
We show that \sysname matches or exceeds the performance of fat trees for a range of traffic patterns, despite being up to 45\% cheaper.
\sysname is now the default datacenter network for most workloads at \cloudone.

\end{abstract}

\maketitle

\begingroup
\renewcommand\thefootnote{\dag}
\footnotetext{Also affiliated with the University of Washington}
\endgroup

\begingroup
\renewcommand\thefootnote{\ddag}
\footnotetext{Also affiliated with the University of California, Santa Cruz}
\endgroup

\vspace{12pt}

\section{Introduction}
\label{sec:introduction}

The simplicity of fat tree topologies has made them the workhorse of datacenter networks. But they offer a stark trade-off between cost and performance. Operators either build non-blocking fabrics in which any endpoint can transmit/receive at its full capacity; or they build oversubscribed fabrics that congest when a small, unfavorable set of endpoints transmit/receive at a high rate. Non-blocking fabrics are prohibitively expensive at scale, so organizations tend to build oversubscribed fabrics and risk congestion even for performance-sensitive workloads~\cite{meta-ml,alibaba-hpn2024,jupiter-rising2015}.

The root
problem is that fat trees lack {\em capacity fungibility}. The strict hierarchical structure means traffic between pairs of endpoints is limited to small subsets of links in the topology. These links can congest while most others lie idle.
The lack of capacity fungibility increases the cost of providing a consistent, congestion-free experience. One must provision ample capacity everywhere even if a small subset of endpoint pairs need high capacity at any given time.

Researchers have proposed reconfigurable topologies for this problem~\cite{projector2016,flyways2011,firefly2014,c-through2010,helios2010,mirror-mirror2012,rotornet2024, opera2020,lightwave2023}. These designs use hardware such as optical switches, steerable wireless antennas, and free-space optical devices to dynamically change capacity between endpoints. Such hardware is not (yet) proven at scale. Most designs also use a centralized control plane to infer fabric-wide demand and configure the hardware. Such control planes are difficult to realize at scale for bursty, latency-sensitive workloads such as Web services. 

A promising alternative that provides capacity fungibility is a flat topology with commodity routers connected as an expander graph~\cite{jellyfish2012,xpander2016,throughput-centric2021,beyond-fat-trees2017,jyothi2016measuring}. 
In addition to their lower cost, such topologies are fault tolerant because the blast radius of any failure is small, unlike upper-layer failures in fat trees. 
These benefits of expander network topologies have been known for over a decade~\cite{jellyfish2012}, but these topologies are not a reality because of three unsolved challenges~\cite{physical-deployability2023,jupiter-rising2015}.

\parab{1. Routing} 
Shortest-path routing, commonly used for fat trees, is a poor fit for expanders because it cannot exploit the diversity of paths that exist. Some pairs of nodes have only one shortest path between them, so shortest path routing can cause congestion. Prior works recommend $k$-shortest paths routing which spreads traffic between the source and destination across $k$ shortest paths~\cite{jellyfish2012,xpander2016}. But $k$-shortest paths routing cannot be realized for large networks using commodity switches. It needs an order of magnitude more (fast, expensive) memory than what is available today (\autoref{sec:requirements}).

\parab{2. Cabling}
Expanders connect pairs of devices that can be far away in the physical space. Such cabling is complex and expensive~\cite{physical-deployability2023}. Worse, when more racks land, the process of incrementally adding new routers requires breaking existing connections. If routers have $d$ links to other routers, each time a router lands, we need to break $d/2$ existing links whose endpoints are physically spread over the datacenter. This type of work risks collateral damage and slows down rack installs, a key concern for datacenter operators. 

\parab{3. Performance predictability} Prior work on benchmarking expanders focus on concrete topologies with fixed parameters. It is unclear how a fabric might perform for other parameters, which makes it difficult for operators to design topologies that meet specific performance targets. They would need to search the combinatorial parameter space using simulations, which is infeasible for large networks. In contrast, fat tree designs can be automatically generated based on performance targets~\cite{fat-tree2008,condor2015,taming2011}.

We solve these challenges and deploy (to our knowledge) the first production datacenter fabrics based on expanders. 
Our design, called \sysname{}, connects routers using a mix of randomized and deterministic cabling segments. This construction yields quasi-random graphs that mimic the statistical properties of truly random graphs~\cite{ChGrWi89}. 
Our routing algorithm called {\em Spraypoint} is purpose-built for random graphs and permits a fully distributed implementation on commodity hardware. 
It finds a large number---close to the node degree---of edge disjoint paths between endpoint pairs and these paths minimally overlap across different pairs. These properties lead to capacity fungibility and high throughput.

We develop a cabling approach based on a new passive optical device called a {\em ShuffleBox} that mixes connections between routers internally. We place shuffle ShuffleBoxes at a small number of planned locations in the datacenter. This allows the number of physically-connected location pairs, a key driver of cabling complexity, to match that of fat trees.

The concentration properties of random graphs and the decorrelation of Spraypoint routing enable accurate modeling of \sysname{} fabric performance as a function of topology parameters (e.g., graph size, node degree, etc.). We develop models for path length, number of edge disjoint paths, and oversubscription. The models make the impact of various parameters and design trade-offs (e.g., path length versus oversubscription) transparent. Operators can use them to design topologies that meet their performance targets.

We deploy two \sysname-based production fabrics.
Transport and application layer benchmarks confirm that \sysname's performance is on par with fat trees, unimpacted by the peculiarities of expanders such as variable path lengths between endpoint pairs.
Such end-to-end validation is a first for expander topologies. 

Our analysis reveals that \sysname topologies are 9--45\% cheaper than fat trees with equivalent oversubscription ratio. The extent of cost reduction depends on the oversubscription ratio and is independent of network size and the number of switch ports. 
We also find that, for the same oversubscription ratio, \sysname{} offers higher throughput than fat trees across a range of traffic patterns.

\section{The potential of \change{flat} expanders}
\label{sec:background}

The salient property of expander graphs is edge expansion~\cite{wikipedia_expander_graph}. For every minority subset $S$ of nodes, the set of edges that lead out of $S$, called the cut, is large.
Edge expansion
 provides capacity fungibility because every $S$ has large bandwidth toward other nodes. In contrast, in hierarchical topologies such as fat trees, the cut of a subtree only leads to parents, so the expansion is poor. We illustrate with an example.

\begin{figure}[t!]
    \centering
    \begin{subfigure}[b]{0.77\columnwidth}                 
        \includegraphics[width=\columnwidth]{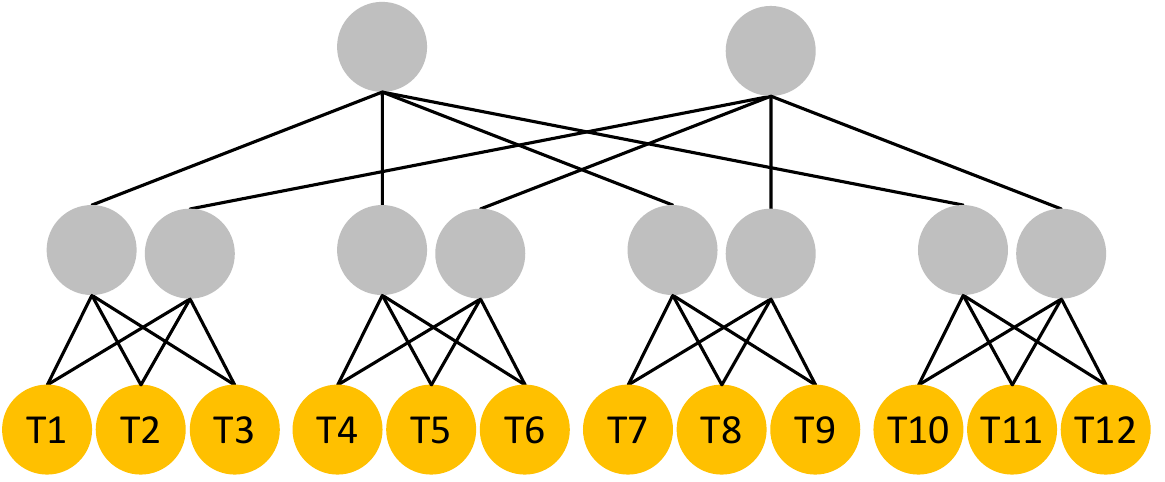}
        \caption{}
        \label{fig:example-hierarchical}
    \end{subfigure}
    \hfill
    \begin{subfigure}[b]{0.22\columnwidth}
        \includegraphics[width=\columnwidth]{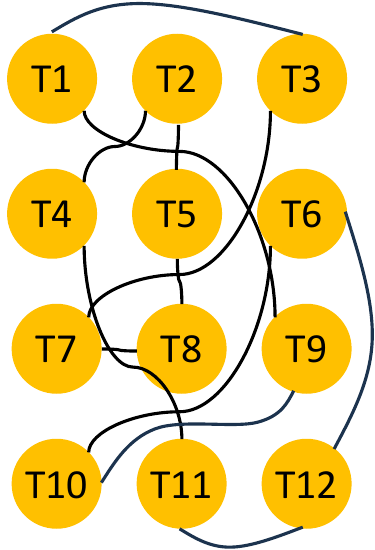}
        \caption{}
        \label{fig:example-random}
    \end{subfigure}
    \caption{Example fat tree (a) and expander topologies (b). Each switch has four ports. Two ports of each ToR (T1--T12) connect to servers (not shown).}
    \label{fig:example}
\end{figure}

\autoref{fig:example-hierarchical} shows a generalized fat tree~\cite{generalized1995}. T1--T12 are ToR (top-of-rack) switches, connected via aggregation (middle) and spine (top) layers. The oversubscription ratio of this network is 3:1 since the capacity of the ToR-aggregation layer is thrice that of the aggregation-spine layer.
The cut for subset of nodes with T1--T3 and their two parents is two links connecting to the spines. These nodes cannot send more than two units of traffic to T4-T12 even though most other links in the fabric are unused. Thus, capacity is stranded.

While skewed traffic patterns with a few heavy transmitters have received the most attention~\cite{beyond-fat-trees2017,jyothi2016measuring}, fat trees can strand capacity with uniform traffic as well.  Consider an all-to-all pattern where each ToR wants to send the same data to each other ToR. The transmission rate between pairs of ToRs will be $\frac{8}{12 \times 9}$ of link capacity. This calculation follows from the $12 \times 9$ flows crossing the $8$ aggregation-spine links. Given two units of per-ToR capacity, this traffic pattern strands roughly 60\% ($2 - 11\frac{8}{12 \times 9}$) of the ToR uplink capacity. 

Expanders, such as the random graph between T1--T12 in \autoref{fig:example-random}, can do significantly better. 
If T1--T3 are the only senders, they will be limited largely by their local capacity when other nodes are silent. Because of the expansion property, there is no small cut that constrains traffic. Even for the all-to-all traffic pattern, given a capable routing algorithm, we can effectively use all links in the topology.

Capacity fungibility lowers cost because the deployed capacity is more efficiently used, allowing smaller deployments.
A simplistic view of the lower cost of flat expanders is that all upper-layer routers are removed.
But an expander that is performance-equivalent to a fat tree may need more ToR uplinks (fabric-facing ports) because some uplink capacity is consumed by traffic relayed for other ToRs. It may thus need more ToRs to support the same number of servers, given a fixed number of ports per ToR. Despite this effect, expanders can reduce cost by up to 45\% (\autoref{sec:eval-cost}).

Expander topologies are also more fault-tolerant. Edge expansion makes it difficult to partition the network, and flat structure lowers the impact of failures.
In a fat tree, \newtext{upper-layer router failures have a large impact because they carry traffic for many endpoint pairs. Even a single spine router failure in \autoref{fig:example-hierarchical} halves the available capacity for most ToR pairs. In \autoref{fig:example-random}, there are no such special routers.}

\newtext{
We focus on multi-tenant datacenters with heterogeneous workloads. In the future, we will investigate using expanders for specialized workloads such as large-scale AI training. These workloads may demand specific topologies such as Rail-optimized fat trees and may want local capacity islands such as those at the aggregation layer in \autoref{fig:example-hierarchical}~\cite{rail-optimized,fat-tree-collectives}. Flat, random topologies do not have aggregation islands.\footnote{\sysname{} does have rack-level islands like fat trees. The lack of aggregation islands is not a problem for general workloads in multi-tenant datacenters. Coordinated application deployment across topologically related racks is hard to realize here; tenants come and go continuously, with each needing different server counts and types. That is why hyperscalers aim for fabric-wide, uniform capacity pools~\cite{vl2}.}}

\section{Requirements and Challenges}
\label{sec:requirements}

Expander topologies are theoretically promising. But to be practical at scale at \cloudone and supplant fat trees, they must meet several requirements. Prior expander-based designs~\cite{jellyfish2012,xpander2016,slimfly2014} do not meet these requirements.

\parab{Realizable using commodity switches} The network should be realizable using commodity switches and forwarding ASICs which have limited memory. While custom hardware can be developed, it increases cost and poses substantial technical and supply-chain risks. 

Prior designs cannot be realized for large networks using commodity switches. 
To counter the limitations of shortest paths in expander graphs, they usually propose $k$-shortest paths routing.
Its typical implementation requires tunnels (e.g., MPLS, VLAN~\cite{spain2010}) or forwarding based on source (in addition to destination) addresses. Suppose we want to build a network with $n$=10K routers and use the recommended value of $k$=8~\cite{jellyfish2012,xpander2016}. If each path traverses $4$ routers, we need 320K forwarding entries per router on average ($4kn^2$ entries spread across $n$ routers). 
Current switches support only 4-16K such entries,\footnote{This table is different from the much larger table for destination-based lookup which uses longest-prefix matching.} which is 20--80x lower. State reduction techniques~\cite{full-label2005,asymmetric-tunnels2005} cannot bridge this gap and adding so much memory is prohibitively expensive. 
\newtext{Harsh et al. cleverly use VRFs (virtual routing and forwarding) to create non-shortest paths, but this proposal too does not scale given limits on the number of VRFs in commodity switches~\cite{spineless2020}.}

\parab{Deployment and operational simplicity} Deployment and operational simplicity are critical for network reliability. 
Simplicity is multi-faceted; we focus on control plane and physical operations to deploy and incrementally expand the network.
For the control plane, we prefer demand-oblivious, fully distributed routing (like today). 
For deployment and expansion, we want the number of cabling steps to be small and fast to minimize risk and expand faster. 

Physical cabling is challenging for expanders because they need to connect pairs of devices uncorrelated to physical space. 
Jellyfish, which is based on truly random graphs, proposes to simplify cabling by removing routers (ToRs) from server racks and placing them centrally~\cite{jellyfish2012}, but this approach increases latency between servers that share a rack, an important consideration for some applications. It also increases cabling cost because cheap, copper-based intra-rack connectivity would need to be replaced with expensive, optical connectivity. 
All current approaches also require changing many existing connections each time a new rack lands, a frequent operation in datacenters.

\parab{Predictable performance} \cloudone operators deploy network fabrics to meet specific capacity (number of servers) and performance (e.g., oversubscription) targets. For expanders to be acceptable, they must be able to easily and confidently create topologies that meet their targets.
Prior work benchmarks topologies with specific parameters and does not show how to create topologies for specific performance targets.

\sysname meets these three requirements via: (1) A new routing algorithm called {\em Spraypoint} that finds a large number of edge-disjoint paths. Like OSPF and BGP, it is demand-oblivious and permits a fully-distributed operation on commodity hardware. 
(2) A cabling approach based on a passive optical device called a {\em ShuffleBoxes}. It limits the number of endpoints that are physically cabled/re-cabled when the datacenter expands. 
(3) High-fidelity models for key performance measures such as throughput and path length. We describe these elements in more detail below.

\section{\sysname{} Overview}
\label{sec:overview}

\sysname{} is based on a flat graph where routers interconnect through a mix of deterministic and randomized choices. 
Given their controlled randomness, our graphs are not truly random; they are quasi-random graphs that behave like 
truly random graphs and are optimal expanders~\cite{Fr08, expansion2014,ChGrWi89}\footnote{Further, unlike structured constructions like Xpander \cite{xpander2016} and Slim fly \cite{slimfly2014}, random graphs can support routers with different degrees~\cite{jellyfish2012}. \protect\sysname{} supports such routers, but we omit this discussion for lack of space.}.
\sysname{}  routers connect to each other and to external devices (e.g., servers, other fabrics). We break out each fabric-facing physical port into individual lanes (e.g., a 400 Gbps port into 4x100 Gbps lanes), each of which uses a separate fiber pair and can establish adjacency with a different remote router. Breakouts increase graph degree, lowering hop count and oversubscription. 
In the rest of the paper, a router uplink refers to a breakout lane.

The control plane and load balancing in \sysname{} follows today's common paradigm. The routing protocol, Spraypoint, computes next hops at each router for all destinations; and routers use ECMP (equal cost multipath) to spread traffic across all next hops for the destination. \newtext{\autoref{fig:params} summarizes the key design parameters of \sysname{}.}

\begin{figure}[t!]
\begin{tabular}{c|c|l}
   \textbf{Graph} & 
   $n$, $d$  &  \# of routers, router uplinks \\ \hline
   \textbf{Spraypoint} &
   $p$, $h$ & \# of waypoints, next hops \\
   & $\ell$ & \# of levels \\ \hline
   \textbf{ShuffleBoxes} &
   $d_r$, $d_c$ & \# of r-ports, c-ports \\
   &
   $f_r$, $f_c$ & \# of fiber pairs in r-port, c-port
\end{tabular}
\caption{Key design parameters of \sysname{}\label{fig:params}}
\end{figure}

\section{Spraypoint routing}
\label{sec:spraypoint}
Spraypoint constructs a large number of edge disjoint paths between endpoint pairs, which increases network throughput by offering more independent options to load balancing mechanisms like ECMP. While we analyze and deploy it on random graphs, Spraypoint works on any expander graph. 
It is based on the observation that high fan-out at the source and high fan-in at the destination suffice for creating many edge-disjoint paths in an expander.  The middle has many edges because of the expansion property.
Spraypoint achieves high fan-out by \emph{spraying} packets at the source---\newtext{all neighbors are eligible next hops and one is selected based on ECMP hashing.} 
It achieves high fan-in by channeling traffic via \emph{waypoints} that are spread all around the destination.

\parab{Forwarding paths} 
Spraypoint has two parameters $p$ and $h$ \newtext{which control the number of waypoints per destination and the number of eligible next hops after the spraying step.}
These parameters offer a trade-off between path length and the number of edge disjoint paths. There is an auxiliary 
parameter $\ell$ which depends on these parameters and the graph size ($n$). 
To compute paths to a destination $t$, Spraypoint partitions all nodes into the following {\em levels}. 
\begin{asparaenum}
    \item $\wayp_0(t)$: Base waypoint level with all neighbors of $t$. 
    \item $\wayp_{l \in [1, \ell]}(t)$: Higher-level waypoints. $\wayp_l(t)$ has $p$ randomly-selected neighbors
    of each node in $\wayp_{l-1}(t)$. Neighbors in earlier waypoint levels and $t$ itself are not eligible for selection.
    \item $\iring(t)$: {\em Inner ring} with all neighbors of $\wayp_\ell(t)$ not in previous levels.
    \item $\oring(t)$: {\em Outer ring} with all nodes not in any level above.
\end{asparaenum}

\smallskip

\begin{figure}[t!]
    \centering
    \begin{tikzpicture}[
    scale=0.5,
    every node/.style={font=\small},
    v/.style={circle, draw=black, fill=white, minimum size=5mm,inner sep=0pt},
    w/.style={circle, draw=black, fill=yellow!50 , minimum size=5mm,inner sep=0pt},
    r/.style={circle, draw=black, fill= cyan!20, minimum size=5mm,inner sep=0pt},
    o/.style={circle, draw=black, fill=lightgray, minimum size=5mm,inner sep=0pt},
    t/.style={circle, draw=orange!80!black, fill=orange!80, text=black, minimum size=5mm,inner sep=0pt},
    edge/.style={thick}
]

\def\rT{0}
\def\rV{1.5}
\def\rW{2.6}
\def\rR{3.8}
\def\rO{5.0}

\draw[gray!40] (0,0) circle (\rV);
\draw[gray!40] (0,0) circle (\rW);
\draw[gray!40] (0,0) circle (\rR);
\draw[gray!40] (0,0) circle (\rO);

\node[t] (t) at (0,0) {$t$};

\foreach \i/\ang in {1/90,2/0,3/-90,4/180} {
    \node[v] (v\i) at (\ang:\rV) {$v_{\i}$};
    \draw[edge] (t) -- (v\i);
}

\foreach \i/\ang in {
    1/90,2/30,3/5,4/-30,
    5/-65,6/-120,7/-170,8/150
} {
    \node[w] (w\i) at (\ang:\rW) {$w_{\i}$};
}

\foreach \i/\ang in {
    1/60,2/40,3/20,4/0,5/-25,
    6/-60,7/-85,8/-105,9/-125,
    10/-145,11/-175,12/170,13/145,14/125,15/100
} {
    \node[r] (r\i) at (\ang:\rR) {$r_{\i}$};
}

\node[o] (o1) at (-15:\rO) {$o_1$};
\node[o] (o2) at (30:\rO) {$o_2$};
\node[o] (o3) at (-90:\rO) {$o_3$};

\draw[edge] (v1) -- (w1);
\draw[edge] (v1) -- (w2);
\draw[edge] (v1) -- (w8);
\draw[edge] (v2) -- (w3);
\draw[edge] (v2) -- (w4);
\draw[edge] (v3) -- (w5);
\draw[edge] (v3) -- (w6);
\draw[edge] (v3) -- (r8);
\draw[edge] (v4) -- (w7);
\draw[edge] (v4) -- (w8);
\draw[edge] (v4) -- (r12);

\draw[edge] (w1) -- (r1);
\draw[edge] (w1) -- (r14);
\draw[edge] (w1) -- (r15);
\draw[edge] (w2) -- (r1);
\draw[edge] (w2) -- (r2);
\draw[edge] (w3) -- (r3);
\draw[edge] (w3) -- (r4);
\draw[edge] (w4) -- (r4);
\draw[edge] (w4) -- (r5);
\draw[edge] (w4) -- (w5);
\draw[edge] (w5) -- (r6);
\draw[edge] (w5) -- (r7);
\draw[edge] (w6) -- (r8);
\draw[edge] (w6) -- (r9);
\draw[edge] (w6) -- (r11);
\draw[edge] (w7) -- (r10);
\draw[edge] (w7) -- (r11);
\draw[edge] (w8) -- (r12);
\draw[edge] (w8) -- (r13);
\draw[edge] (w8) -- (r14);

\draw[edge] (v2) -- (o1);
\draw[edge] (r6) -- (o1);
\draw[edge] (o1) -- (o2);
\draw[edge] (r2) -- (o2);
\draw[edge] (r3) -- (o2);
\draw[edge] (w5) -- (r6);
\draw[edge] (r7) -- (o3);
\draw[edge] (r8) -- (o3);

\begin{scope}[xshift=6cm,yshift=1cm]
\node[v] at (0,2.4) {$v_i$};
\node[right] at (0.5,2.4) {$WP_0(t)$};

\node[w] at (0,1.2) {$w_i$};
\node[right] at (0.5,1.2) {$WP_1(t)$};

\node[r] at (0,0) {$r_i$};
\node[right] at (0.5,0) {$IR(t)$};

\node[o] at (0,-1.2) {$o_i$};
\node[right] at (0.5,-1.2) {$OR(t)$};
\end{scope}

\end{tikzpicture}
    \caption{Example graph with Spraypoint levels for $t$.}
    \label{fig:rings}
\end{figure}
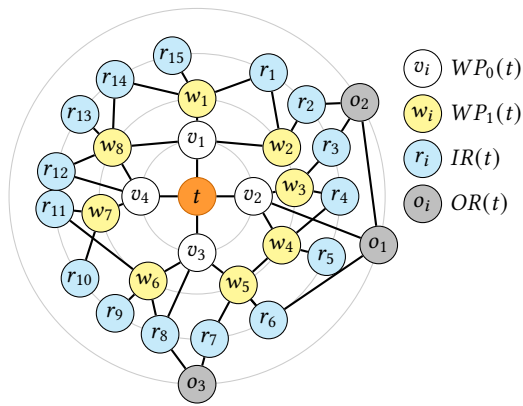

\autoref{fig:rings} shows an example graph for $\ell$=1 and $p$=2. $\wayp_0(t)$ has all neighbors of $t$: $\{v_1, v_2, v_3, v_4\}$. $\wayp_1(t)$ has two randomly selected neighbors of each $v_i$. 
In the example, $w_1, w_2$ are selected from $v_1$, $w_3, w_4$ are selected from $v_2$, and so on. Node $v_1$ ends up with three neighbors in $\wayp_1(t)$, because an additional neighbor ($w_8$) was selected via $v_4$.
The neighbors of $v_i$ that are not selected in $\wayp_1(t)$ drop into the rings, along with neighbors of $\wayp_1(t)$. 

The levels control traffic flow between any source $s$ to $t$.
The first step is \emph{spraying}, where $s$ forwards to a randomly chosen neighbor (i.e., ECMP). Spraying is independent of $t$ and happens even if $s$ and $t$ are adjacent.\footnote{Spraying is reminiscent of Valiant Load Balancing~\cite{vlb1981}. But unlike VLB, which uses arbitrary nodes as intermediaries, it uses only neighboring nodes. Its primary goal is high fan-out at the source, not load balancing, though high fan-out does aid load balancing.}
Post spraying, each intermediate node $u$ follows the "pointing" rules. 

\begin{asparaenum}
    \item If $u \in \wayp_0(t)$, forward to $t$.
    \item If $u \in \wayp_{l \in [1, \ell]}(t)$, forward (ECMP) to one of $h$ neighbors randomly selected from $\wayp_{l-1}(t)$. 
    \item If $u \in \iring(t)$, forward to one of $h$ neighbors randomly selected from $\wayp_\ell(t)$. 
    \item If $u \in \oring (t)$, forward to one of $h$ neighbors randomly selected from those with shortest paths to $\iring(t)$ nodes. 
\end{asparaenum}

\smallskip

Let us see Spraypoint forwarding in action for a packet from $v_2$ to $t$ in \autoref{fig:rings}. 
Assume $h=1$.
First, $v_2$ sprays to one of $t$, $w_3$, $w_4$, or $o_1$. If the packet is sent to $t$, its journey is complete.
If it is sent to $w_3$ or $w_4$, it will go back to $v_2$ per the second rule above, and then to $t$. Spraying happens only at the source; if a packet happens to return to the source, it follows the pointing rules. If the packet was sprayed to $o_1$,
it reaches $t$ via $r_6$, $w_5$, and $v_3$.

\parab{Why spraying} Spraying exploits the expansion property to generate many paths between source-destination pairs. 
In general, there may be only a few {\em shortest} paths from a source $s$ to a destination $t$. But in an expander, there are many edge disjoint \emph{short} paths between all pairs. A simple heuristic to use these paths is to follow shortest paths to $t$ from all neighbors of $s$. Spraying implements this heuristic.

\parab{Why waypoints} Spraying alone fails in some cases. Suppose $s$ is a neighbor of $t$. After $s$ sprays, its neighbors use shortest paths to $t$. Unfortunately, unless a neighbor is directly connected to $t$, which will be rare in a large graph, the shortest paths
go via $s$. So almost all neighbors route the traffic back to $s$, and the $s$$\rightarrow$$t$ link will congest. Higher-level waypoints ($\wayp_{l \in [1, \ell]}(t)$) draw out traffic further, preventing a collapse on this link. Further, selecting $p$ of $d$ nodes builds a $p$-ary forwarding graph rooted at each neighbor of $t$, which helps spread traffic across all neighbors (high fan-in). Otherwise, some neighbors may not be used. 

\parab{Setting $\ell$}
To achieve the goals of spraying and waypointing, we set $\ell$=$\max(1, \lceil \log_p(n/2d^2) \rceil)$.
The size of the set $\wayp_\ell(t) \cup \inner(t)$ will be at least $n/2$ for this choice of $\ell$ (\autoref{sec:modeling}), which reduces path lengths because each node in $\oring(t)$ will connect directly to this set with high probability. 

\parab{Path length variability} Like other proposals for routing in expander graphs~\cite{spineless2020,jellyfish2012,xpander2016}, Spraypoint computes paths of different lengths between a pair of endpoints. 
This variability will not impact single-path transport protocols like TCP which sample only one path (based on ECMP), but it could confuse multipath protocols~\cite{mptcp,srd,falcon} that spread traffic across multiple paths (by varying packet headers used by ECMP) if they use latency differential as a congestion or load balancing signal. 
In a datacenter that spans 300 meters and uses 100 Gbps links, a 2-hop path length difference creates a maximum latency difference of 4.4 $\mu s$ (0.7 $\mu s$ transmission time for a 9 KB packet, 1.5 $\mu s$ propagation delay per hop). This differential is low compared to endhost latencies and a small amount of queuing (6 packets) will wipe it out.
Benchmarking on production fabrics confirms the absence of performance issues for multipath protocols (\autoref{sec:implementation}). 

\parab{Distributed implementation} 
Spraypoint permits a fully distributed implementation as a link-state protocol, where every node has a full view of the topology. Nodes can compute levels for each destination and then compute local next hops by applying the pointing rules. All nodes make identical waypoint selections, via deterministic hashing of a shared key, ensuring that all pointing graphs are loop free.

We implement spraying using VRFs (virtual routing and forwarding), which enable different routing logic for different interfaces. 
\newtext{Unlike prior work that uses VRFs to create non-shortest paths~\cite{spineless2020} (which does not scale), we need only two VRFs.} 
All server-facing interfaces of a router use a VRF that sprays incoming traffic to all fabric-facing interfaces. (One exception is traffic destined to servers connected to the same router, which is forwarded directly to those servers.)  All fabric-facing interfaces use a second VRF that follows the pointing rules. 
This setup allows a packet to revisit its source (at most once) without looping indefinitely.

\parab{Resource requirements} Spraypoint uses two types of router memory. The first is longest-prefix matching (LPM) memory, which maps the destination address of an incoming packet to an ECMP group id. LPM memory use depends on the number of network prefixes and is similar to fat trees.\footnote{An exception is when hierarchy is leveraged to aggregate some prefixes. Such aggregation is not feasible in flat expanders. Because LPM memory is huge, it can support even the largest \cloudone DCs without aggregation.} 

The second type of memory is to map ECMP group id to next hop set.
Its required size depends on how ECMP groups are setup. 
The two possibilities are: (1) a separate ECMP group for each destination node, which consumes $O(nh)$ memory;
and (2) pre-define all possible $d^h$ groups and use the corresponding one for each destination, which consumes $O(hd^h)$ memory. Based on $n$ and $d$, we pick the method that supports the larger $h$ while fitting in the memory. 
For 128-port switches, the value of $d$ is around 64, so a minimum of $h$=2 is always feasible independent of fabric size ($hd^h \approx 8K)$. 

The computational complexity of Spraypoint is $O(n^2d)$ per node because each node inspects $nd$ edges per destination to compute levels. The CPUs of modern commodity switches can handle such complexity~\cite{dsdn}. For comparison, the complexity of k-shortest-paths routing is $O(kn^2(d + \log n))$~\cite{wikipedia_k_shortest_paths}.

\section{Cabling quasi-random graphs}
\label{sec:cabling}

We explain how quasi-random graphs are realized in \sysname after a brief background on datacenter cabling.

\parab{Datacenter cabling} 
Datacenter space is typically partitioned into O(10) {\em rooms} (not necessarily with dividing walls). Each room hosts some number of racks. 
Before server racks land, the room is {\em prepared} by installing power, cooling, and basic cabling infrastructure including patch panels and inter-room trunk cables. 
When server racks land, additional cables are installed to connect them. 

The cost and complexity of cabling stems from three main factors: (1) the number of connectors and endpoint pairs, where a patch panel is an endpoint and so is a ToR; (2) the length of cables and the fraction that are intra- versus inter-room---inter-room cabling is more expensive because it needs extra infrastructure; and (3) the amount of cabling that can be done during room preparation versus when racks land---later cabling is riskier and slower.

\begin{figure}[t!]
    \centering
    \includegraphics[width=0.5\columnwidth]{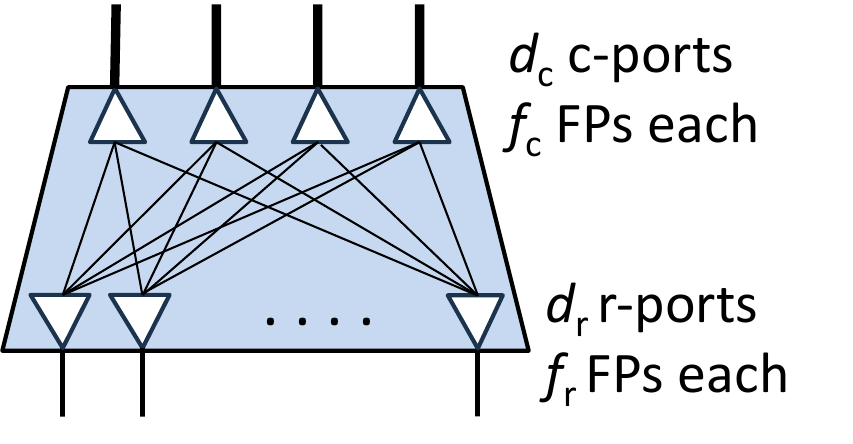}
    \caption{A ShuffleBox.}
    \label{fig:ShuffleBoxes}
\end{figure}

\parab{Physical connectivity in \sysname{}} 
\label{sec:physical-dublin}
We build a \change{quasi-}random graph using ShuffleBoxes (\autoref{fig:ShuffleBoxes}).
\newtext{Each ShuffleBox has $d_r$ {\em r-ports} that connect to routers and $d_c$ {\em c-ports} that connect to other ShuffleBoxes.} 
R-ports and c-ports terminate, respectively, $f_r$ and $f_c$ fiber pairs (FPs) each. An FP enables duplex communication between routers. ShuffleBoxes connect FPs between r-ports and c-ports ($d_r{\times}f_r{=}d_c{\times}f_c$). We describe these parameter values later. 

\begin{figure}[t!]
    \centering
    \includegraphics[width=0.8\columnwidth]{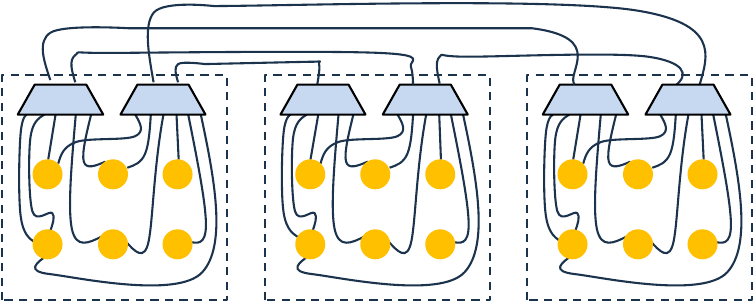}
    \caption{Cabling in a datacenter with three rooms. Routers connect to r-ports of ShuffleBoxes in the room, and ShuffleBoxes inter-connect via c-ports.}
    \label{fig:datacenter}
\end{figure}

To connect routers via ShuffleBox, we deploy a set of ShuffleBox, called a {\em shuffle panel}, in each room (\autoref{fig:datacenter}). Router uplinks connect to randomly selected r-ports in the panel.
The c-ports of panels in different rooms are also randomly connected. If there are $R$ panels, each with $C$ ShuffleBoxes, we connect approximately $\frac{(R-1)Cd_c}{R}$ c-ports of each panel to other panels, essentially building a random graph between panels.  
Unconnected c-ports of a panel have a {\em ShuffleBack}, a special connector that bridges pairs of FPs coming into the c-port (from r-ports). It bridges FP1 to FP2, FP3 to FP4, and so on. 
Panel r-ports that are not connected to routers also have a ShuffleBack that bridges c-port FP pairs. 
\newtext{Routers connect only to the ShuffleBoxes (never directly), and topology modifications are achieved by changing only connections between the ShuffleBoxes (as explained below).}

\autoref{fig:connectivity-patterns} shows three  ways in which two routers can logically connect: (a) routers in the same room connect to the same ShuffleBoxe and their FPs are shuffled to a c-port with a ShuffleBack;\footnote{With random mapping of router uplinks to ShuffleBox r-ports, there is a chance that two uplinks of the same router land on the same ShuffleBox and are bridged via a c-port ShuffleBack. For realistic router counts, the probability of self-edges is negligible since each ShuffleBox is small.} 
(b) routers in different rooms whose FPs are shuffled to connected c-ports; (c) routers whose FPs end up at the same ShuffleBacked r-port in another room. In the final pattern, the two ToRs may be in the same room.

\begin{figure}[t!]
    \centering
    \includegraphics[width=0.7\columnwidth]{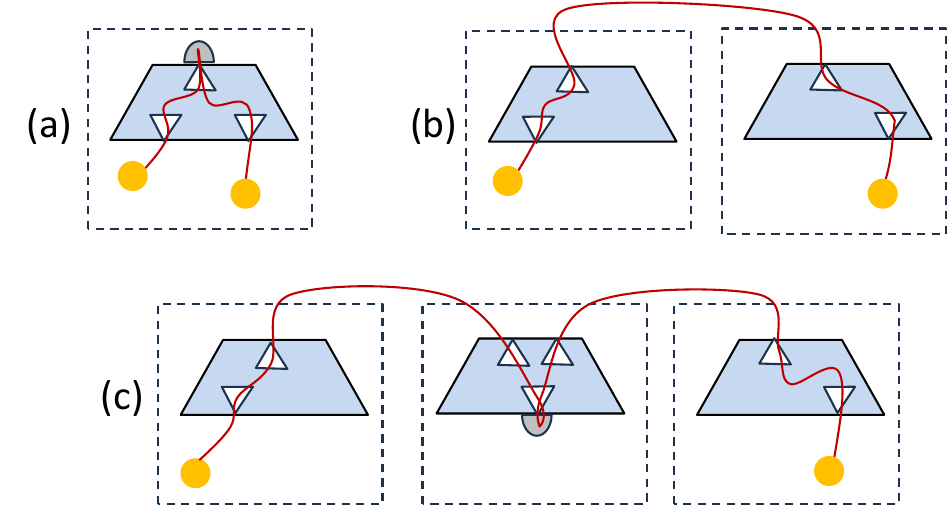}
    \caption{Physical connectivity patterns that enable logical connectivity. The hats denote ShuffleBacks.}
    \label{fig:connectivity-patterns}
\end{figure}

Connecting via ShuffleBoxes has higher optical loss than direct connectivity, but this loss is within the margins of commodity transceivers. Our construction can produce even longer end-to-end paths, e.g., if the right r-port in \autoref{fig:connectivity-patterns}c had a ShuffleBack instead of connecting to a router, the path will extend and go back out via a c-port. 
We disable paths with over seven connectors to maintain optical signal quality.

\parab{Incrementally deploying rooms and racks} \sysname's physical connectivity can be done incrementally. \autoref{sec:app:incremental} describes the process. To summarize the key steps: (1) install a new shuffle panel when a new  room is prepared; (2) if it is not the first room, rebalance inter-panel connectivity by removing some existing (randomly selected) c-port connections and ShuffleBacks and using the opened c-ports to connect to the new panel; and (3) when racks land in the room, connect router uplinks to randomly selected r-ports.

\parab{ShuffleBox configuration} 
We use $f_r$=4, the number of breakout lanes for a 400 Gbps physical port. So that each FP coming into an r-port from a router can reach different places via different c-ports, $d_c{=}f_r{=}4$. We use $f_c$=32 to balance availability of commodity connectors (larger values are less common) and effort to rebalance inter-panel connectivity (small values increase effort). Finally, since $d_r{\times}f_r{=}d_c{\times}f_c$, we get $d_r{=}32$. 
With this configuration, the shuffling pattern of ShuffleBoxes is simple: each c-port FP goes to a different r-port, and each r-port FP goes to a different c-port (full bipartite).

\parab{Cabling complexity} 
\sysname{}'s cabling has low complexity per the three factors outlined earlier. With $n$ routers and $R$ rooms, the number of physically-connected endpoint pairs is $n + \frac{R (R -1 )}{2}$ as opposed to $n^2$ logically-connected pairs. We count each shuffle panel as one endpoint because after the trunk cables reach the panel, the intra-panel distribution is simple. Inter-room cables are limited to $\frac{R (R -1 )}{2}$ trunks between shuffle panels. Other than router to r-port connectivity, all cabling can be done before racks land. 
Along these factors, cabling complexity of \sysname{} is on par with fat trees.
Rebalancing inter-panel connectivity is a unique activity in \sysname{}, but it is needed only a small number of times, when new rooms are prepared. 

\parab{Cost} ShuffleBoxes and ShuffleBacks are novel passive optical components  that \sysname uses to simplify cabling. Being passive, their cost is low relative to switches and transceivers and similar to traditional patch panels~\cite{example_patch_panel} and loopbacks~\cite{example_loopback} that are often used to simplify fat tree cabling. 

\parab{Resulting graphs} 
\change{The logical graph built using our physical construction is an optimal expander despite using deterministic router-to-ShuffleBox connections and constraining randomness (since c-port edges are balanced between rooms and bundle edges between multiple router pairs.} 
Because each bundle has random routers and c-port connections are random, our analysis confirms that \sysname topologies have the same spectral gap~\cite{wikipedia_expander_graph}, a measure of expansion, as unconstrained random graphs.

\section{Modeling Performance}
\label{sec:modeling}

\change{Following the standard practice for quasi-random graphs, our analysis assumes that the graphs are truly random.}
The combination of random graphs and spraying-induced decorrelation enables modeling of  \sysname's performance. Our models inevitably make simplifying assumptions but predict performance quite accurately (\autoref{sec:evaluation}). We summarize the key results in this section and include details in the appendix. Our models give approximate formulas that give accurate
predictions of the actual quantities. For convenience, we do not give formal statements here, and defer all
precise mathematical statements and proofs to the appendix.

We consider operating regimes in which $(i)$ $2(\ln(n) + 5) \leq d \ll n$, $(ii)$ $p \geq (n/d^2)^{1/\ell}$, and $(iii)$ $h \ll d$, where $\ell$ is the number of waypoint levels in Spraypoint. (We used the second criterion in \autoref{sec:spraypoint} for setting $\ell$). Performance is predictable in this regime because distributional tails due to randomness are trimmed. 
At the same time, this regime is broad. Since commodity switches support at least 128 breakout lanes~\cite{tomahawk3,tomahawk5}, we can easily pick values of $d$ that exceed the stated lower bound ($\approx 28$ for $n$=10K) for even large fabrics, and we can pick compliant values of $p$ and $h$ . Further, $\ell$=1 suffices for $n$=10K and typical values of $d$ and $p$. Our oversubscription model focuses on that regime for simplicity. 

\subsection{Edge disjoint paths}
\label{sec:model:edge-disjoint}

Spraypoint exploits the expansion property to compute many edge disjoint paths. The formal statement
and proof are given in \Sec{mincut}.

\begin{model} \label{mod:edge-disj} With high probability, the number of edge-disjoint paths between $s$ and $t$ are approximated by:
\[
\begin{cases}
d(1-\exp(-h)) & \text{if }s \notin \wayp_0(t) \\
\min[d-p, d(1-\exp(-(1-p/d)h))] & \text{if }s \in \wayp_0(t) \\
\end{cases}
\]
\end{model}

\noindent
A detailed explanation of this model is in \autoref{sec:mincut},
but we summarize here.
Consider $s \notin \wayp_0(t)$ (i.e., s is not a neighbor of $t$), and spraying a single packet to every neighbor of $s$. The number of edge disjoint paths is closely related to the number of neighbors of $t$ reached by these packets. Modeling this process as a balls and bins calculation~\cite{wikipedia_balls_into_bins}, where each ball is thrown $h$ times, leads to the formula above. 

When $s \in \wayp_0(t)$, $p$ neighbors of $s$ are in $\wayp_1(t)$. When traffic is sprayed by $s$, the traffic 
to these waypoints potentially comes back to $s$. For the other $d-p$ neighbors, we get a version of the balls and bins game. 

\parab{Takeaways} (1) For both types of sources, the reduction in the number of edge disjoint paths compared to the maximum possible value of $d$ is proportional to $\exp(-h)$. Thus, a low value of $h>1$ suffices for good performance. The change in $\exp(-h)$ from $h$=1 (0.37) to $h$=2 (0.13) is large, but it tapers off rapidly as $h$ increases.  (2) For neighboring sources, high values of $p$ reduce the number of edge disjoint paths. 

\subsection{Path length}
\label{sec:path-length}

Path lengths in \sysname{} depend on the sizes of various Spraypoint levels. A packet is sprayed to a random node, and the number of pointing hops depends on the node's level. We can bound the resulting path length distribution. 
\begin{model}
\label{mod:path-length} With high probability, the fraction of paths of length $i$ is approximated by:
\[
\begin{cases}
1/n  & i=1\\
p^{i-2}d/n & 2 \leq i \leq \ell+2 \\
\exp(-p^\ell d^2/n) & i = \ell + 4\\
\textrm{rest} & i = \ell + 3\\
\end{cases}
\]
\end{model}

\begin{figure}[t!]
    \centering
    \includegraphics[width=\columnwidth]{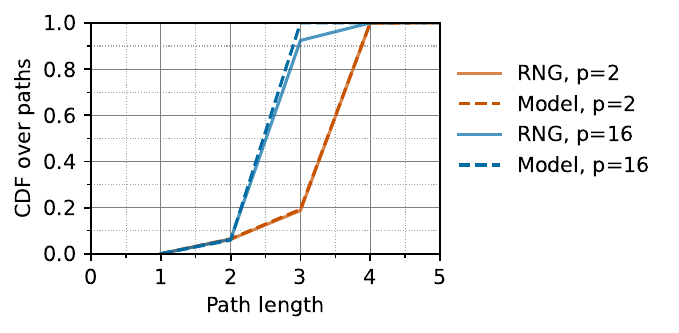}
    \caption{Path length distribution.}
    \label{fig:path-length}
\end{figure}

\Sec{level} has the proof, but \autoref{fig:path-length} compares the model to a simulated fabric. It considers two values of $p$, with $n$=1K and $d$=64. The model matches the simulation well, especially for $p=2$. When $pd > n$, as for $p$=16, lower-order terms that determine level sizes begin to matter. We ignore them for simplicity, but our analysis can be extended if needed. 

\parab{Takeaways} (1) When $\ell=1$ (practical regime), the maximum path length is $5$, and $5$-hop paths are a negligible fraction (e.g., $\exp(-pd^2/n) $ $\approx 0.0003$ for $n$=1K, $d$=64, $p$=2). (2) In this regime, the average path length of \sysname is less than that of 3-tier fat trees, where the vast majority of ToR-to-ToR paths have 4 hops. (3) \sysname paths have variable lengths even across the {\em same} endpoint pair, unlike fat trees where all paths between two endpoints are equal. (4) $h$ does not impact path length. 

\subsection{Throughput (oversubscription ratio)}
\label{sec:model:oversub}

A key measure of a network's throughput is oversubscription ratio. 
It is the minimum fraction of traffic that the network can deliver, without exceeding link capacities, across all possible traffic matrices. 
Unlike for fat trees, one cannot determine oversubscription for expander topologies through a simple analysis of the graph structure. 

Our oversubscription model uses two observations. First, the worst case traffic matrix is a {\em matching} where every node sends at full rate to exactly one other node and receives at full rate from exactly one other (possibly different) node~\cite{throughput-centric2021}. 
We consider a matching with randomly selected source-destination pairs, which can be viewed as modeling {\em stochastic} oversubscription. \newtext{This random traffic matrix might not be the worst case. While matchings with long paths can help find worse traffic matrices ~\cite{jyothi2016measuring,traffic-oblivious2011}, spraying in \sysname{} decorrelates path lengths.}

The stochastic analysis is common in random graph theory---argue that worst case for graph properties (e.g., cuts, expansion, matchings) is close to the random case, using a union bound argument~\cite{MoRa-book, FrMe12}. 
Second, the network throughput is maximized for a given traffic matrix, or the oversubscription ratio is minimized, when traffic takes the shortest possible paths (\Lem{oversub}).

We estimate oversubscription by estimating $\mu_i$ for each flow in a random matching, where $\mu_i$ is the maximum fraction of flow carried by $i$-hop paths. The details are in \autoref{sec:oversub-model}; we summarize below.

We consider $i \in [2, 5]$ because the probability of length 1 paths is negligible and the maximum path length is 5 when $\ell = 1$ (practical regime). We approximate $\mu_2$ as $d/n$
because each flow has $d/n$ fraction of 2-hop paths (\Mod{path-length}). With an assumption that paths do not overlap, we can show that almost every such path carries a single unit of flow.

Estimating $\mu_3$ is more challenging because we cannot assume that paths do not overlap. 
We first estimate $\phi_3$, the fraction of a source $s$'s capacity along 3-hop paths after removing the $d/n$ edges used by 2-hop paths, using the facts that 3-hop paths are of the form $s\rightarrow\wayp_1(t)\rightarrow\nbr(t)\rightarrow t$ and there are $pd$ waypoints. 
We then estimate $\kappa_3$, the amount of source-destination flow that such a path can carry, using the distribution of path overlaps.
\begin{eqnarray}
    & \phi_3 & = \min(pd/n,1-d/n) \cdot (1-d/n) \cdot (1-(4d/n)^h) \nonumber \\
    & \kappa_3 & = (1-\phi_3)^6/2 + (1-\phi_3^2)^3/6 +1/3 \nonumber
\end{eqnarray}

After similar polynomial calculations for $\mu_4$, $\mu_5$, we get:

\begin{model}
\label{mod:oversub} 
The oversubscription ratio is approximated by $(\mu_2 + \mu_3 + \mu_4 + \mu_5)^{-1}$. 
\begin{eqnarray}
\mu_2 & = & d/n \label{eq:mu2} \nonumber \\ 
\mu_3 & = & \phi_3 \kappa_3 \ \ \text{(see above)} \label{eq:mu3} \nonumber \\
\mu_4 & = & [1 - (p+1)d/n - \exp(-pd^2/n)] \nonumber \\
& & \cdot \Big(1 - [1- (1-2d/n)(1-(4d/n)^h)]^h\Big) \nonumber \\
& & \cdot (1-\mu_2 - 2\mu_3)/4 \label{eq:mu4} \nonumber\\
\mu_5 & = & \exp(-pd^2/n)(1-\mu_2 - 2\mu_3 - 3\mu_4)/5 \nonumber \label{eq:mu5}
\end{eqnarray}

\end{model}

This model characterizes the mesh-layer (between routers) oversubscription. For an end-to-end (server-to-server) analysis, oversubscription at the ToRs must also be factored (\autoref{sec:fabric-planning}).

\parab{Takeaways} Fabric oversubscription is a complex interplay of all parameters. By capturing it for a wide range of parameter values,\footnote{The formulas can be simplified for specific regimes, e.g., when a parameter is fixed. For $h$=2, $\log_d(n/p) + 2$ approximates \Mod{oversub} well.}  our model enables fabric design for specific performance targets, which we discuss next.

\begin{omit}
\subsection{Incremental expansion}
\label{sec:model:incremental}

Our model of datacenter expansion predicts the average degree of the fabric at each point in time. We model the expansion as occurring in \emph{stages}, where each stage corresponds to the landing of new panel of ShuffleBoxes and associated re-cabling. Routers of a stage only connect to ShuffleBoxes of that stage. A stage could be a room or a phase within a room. The model makes no assumption about the number or size of stages used to grow the datacenter.

We use ``time" $t \in [0,1]$ as the primary variable to model datacenter growth. It represents the fraction of routers currently deployed, relative to the full DC. $[t_1, t_2]$ denotes a stage that begins at $t_1$ and ends at $t_2$.
Abusing notation, we denote a router by its timestamp of arrival. 
The average number of uplinks per router is $d$.

For a growing network's average degree, we can prove:

\begin{theorem} \label{thm:hyp} The average degree during the stage $[t_1, t_2]$ is $d \Big[ \frac{t_1}{t} + \frac{t-t_1}{t_2}\Big]$
\end{theorem}
The details are in \autoref{sec:appendix:incremental}, but informally, first consider the first stage. For this stage, $t_1$=0 and the formula simplifies to $d\frac{t}{t_2}$, which reflects that the probability of a ToR uplink finding a match increases linearly because the probability of bridging to an occupied r-port grows linearly. 

For subsequent stages, observe that at stage boundaries, the graph degree is $d$ because all routers would have found a match. During a stage, all routers from previous stages maintain their degree of $d$, and routers of the current stage find a match with probability $t/t_2$. The weighted average of the two types of routers yields the model above. 

\parab{Optimal phasing} When datacenter operators use phases to counter the poor early performance of room 1, the growth model can help determine the optimal number and size of phases. We formulate this problem as: Find the minimum number of phases and their sizes such that the average graph degree is at least $\alpha d$ when at least a $\beta$-fraction of the first room is deployed. $\alpha$ and $\beta$ are user-provided parameters. 

To solve this problem, we consider increasing number of phases and find the minimum number that can satisfy the constraint. Based on the growth model, we derive that the minimum degree in a stage $[t_1, t_2]$ is $d(2\sqrt{t_1/t_2} - t_1/t_2)$. This derivation helps us check if the constraint can be met and derive the sizes of the phases. See \autoref{sec:appendix:incremental} for details.

The analysis reveals that we can get average degree to be $0.8d$ as soon as 25\% of the routers of the first room have landed with two phases. The first phase should be 30\% of the routers and the second 70\%.
\end{omit}

\subsection{Designing fabrics}
\label{sec:fabric-planning}

While the exact procedure for determining the topology and routing parameters ($n$, $d$, $p$, and $h$) depends on desired trade-offs (e.g., between oversubscription and path length), we illustrate for a common target: minimize cost (fewest ToRs) of connecting $s$ servers with an end-to-end oversubscription ratio below $r_e$ and ToR-layer oversubscription below $r_t$ ($1{\leq}r_t{\leq}r_e$). Since ToRs are a localized congestion risk, operators usually want to control their oversubscription~\cite{alibaba-hpn2024}.

To meet the target above, we binary search for the minimal viable value of $d$ such that $d \geq \lceil \frac{P}{r_t + 1} \rceil$, $d \geq 2 ln(\lceil \frac{s}{P-d} \rceil) + 5)$, and $d < P$, where $P$ is the number of router ports.
The first constraint ensures that ToR layer oversubscription, which is $\frac{P-d}{d}$:1, is below $r_t$; the second ensures that the design is within the modeled operating regime; and the third ensures that at least one port is available for servers. 
Minimizing $d$ subject to these constraints minimizes ToR count because it maximizes $P$-$d$, the ports that connect to servers.

The viability of a value of $d$ is determined as follows. Based on $d$, we compute the number of required ToRs $n$ as $\lceil \frac{s}{P-d} \rceil$. We then determine the maximum value of $h$ that ToRs can support based on $n$, $d$, and the ECMP memory size (\autoref{sec:spraypoint}). 
For a given $d$, $n$, and $h$, the valid range of $p$ is $[n/d^2, d]$, and different values yield different oversubscription ratios. We compute the range of oversubscription ratios using \Mod{oversub} for extreme values of $p$, and deem $d$ as viable if $r_e/r_t$, which is the desired mesh-layer oversubscription, falls in this range.

Once we have the minimal viable $d$, and corresponding $n$, $h$, we use the maximum $p \in [n/d^2, d]$ with oversubscription less than $r_e/r_t$. Maximizing $p$ minimizes oversubscription.

\section{Production Fabrics}
\label{sec:implementation}

We have deployed the \sysname{} design in two fabrics that carry production workloads.
The first, called the "server mesh", connects ToRs as an expander. The second, called the "edge mesh," connects to the server mesh and to remote datacenters, and it provides transit between these networks. Ports that interconnect the two \sysname{} fabrics are spread across their routers.
We build separate fabrics because they have different oversubscription targets. The  edge mesh is non-blocking and server mesh has the same oversubscription ratio as \cloudone's latest fat tree fabrics.
We elide the oversubscription level and mesh sizes for confidentiality.

\begin{figure}[t!]
    \centering
    \includegraphics[width=\columnwidth,keepaspectratio]{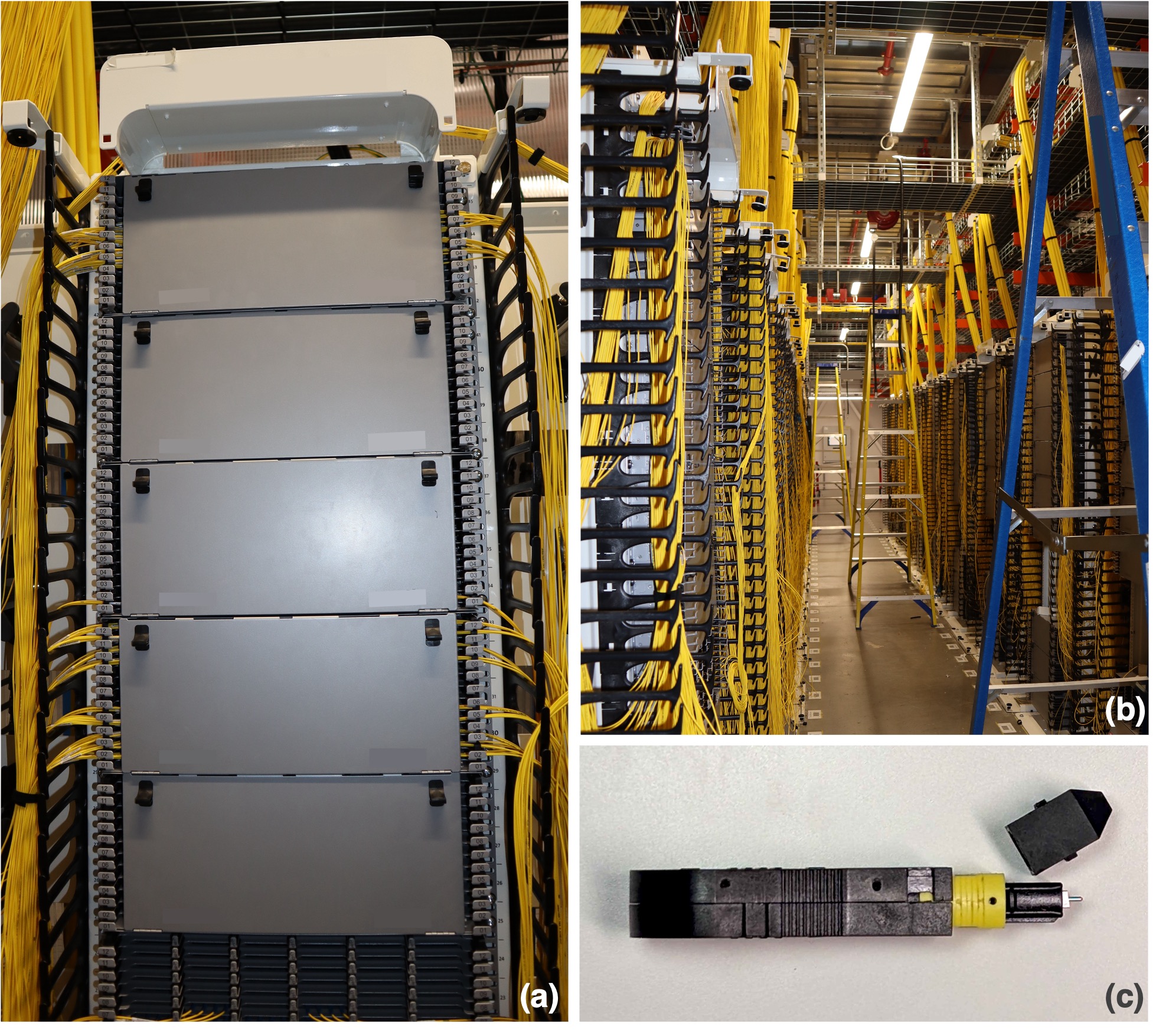}
    \caption{(a) One of the racks hosting the emulated ShuffleBoxes. (b) Rows of emulated ShuffleBoxes. (c) A ShuffleBack dongle with 4-FPs MPO connector.}
    \label{fig:combined_dc}
\end{figure}

At the time of deploying these fabrics, we did not have manufactured ShuffleBoxes. We emulated shuffling with a traditional patch panel that bridges individual FPs and custom ShuffleBacks. See Figure~\ref{fig:combined_dc}.

\parab{Spraypoint performance} We implemented Spraypoint by extending \cloudone's shortest-paths based link-state protocol. \newtext{We reused the topology dissemination component and modified next hop computation. For similarly-sized topologies, Spraypoint performs similarly to the current protocol along key metrics such as convergence time after a failure.}

\parab{Cabling experience} \newtext{Cabling \sysname{} using the scheme in \autoref{sec:cabling} was relatively smooth. One challenge was ensuring enough ShuffleBoxes are installed in each room. The exact number of racks, and thus router uplinks, that can land in a room varies, depending on their power consumption. We used an estimated upper bound on the number of router uplinks per room.\footnote{If this bound were breached, we plan to treat the remaining racks as belonging to a new logical room and perform the room-addition activity. Analogous provisioning challenges arise also for fat trees in large datacenters.}} 

Due to a lack of structure and physically identifiable patterns, a concern with random graphs has been operators connecting ports not intended to be connected~\cite{jellyfish2012,xpander2016}. The incidence of such miscabling in \sysname{} was below 1.5\%.

\parab{Operational challenges} \newtext{We give a short preview of these challenges, even though it is not a focus of this paper.} 
Operationalizing random graphs required upgrading many software tools that manage \cloudone's networks. The hierarchy of fat trees was embedded deeply into these tools, starting from device names itself to managing redundancy during maintenance. For instance, ToRs in fat trees do not rely on each other and can be upgraded based on concerns local to the rack, but in \sysname{} we need to account for inter-ToR dependencies and cannot simultaneously upgrade too many neighbors of any given ToR. \newtext{We ensured that fault localization functions properly and built new tools to easily determine the paths between ToRs for troubleshooting.}

\parab{Application performance}
A majority of traffic in hyperscale datacenters is moving toward multipath transport protocols~\cite{srd,mptcp,falcon}. The sender splits data across tens of "flowlets" that use different paths in the network (based on flowlet hash). The receiver assembles and orders the data, transparent to the applications. The sender uses latency as one of the signals when picking which flowlet to use.

Our main benchmarking goal is to validate that latency differential across sender-receiver paths (\autoref{sec:path-length}) does not impact performance. We study this by observing the raw throughput of the transport protocol and the performance of a latency- and throughput-sensitive block-storage application that uses the protocol.
We compare \sysname{} against a production fat tree fabric with identical oversubscription and identical server specs. We normalized results by the mean values observed for the fat tree to preserve confidentiality. 

\begin{figure}[t!]
    \centering
    \includegraphics[width=0.7\columnwidth]{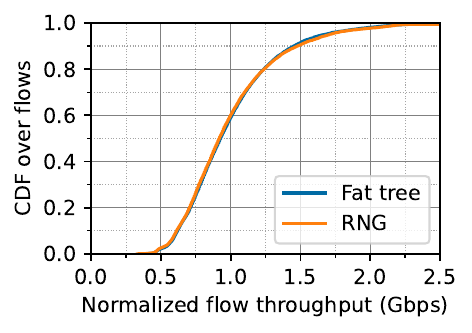}
    \caption{Per-flow throughput of multipath transport (normalized by the fat tree mean).}
    \label{fig:prod_flowthroughput}
\end{figure}

\begin{figure}[t!]
    \centering
    \includegraphics[width=0.9\columnwidth]{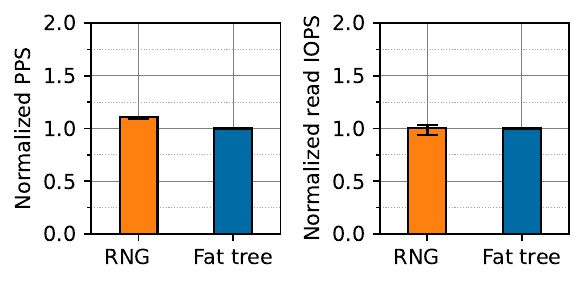}
    \caption{PPS of 64B packets (left) and IOPS of storage reads (right). The error bars are standard deviations.}
    \label{fig:prod_combined}
\end{figure}

\autoref{fig:prod_flowthroughput} shows the throughput distribution based on an experiment with 127 concurrent flows with infinite demand between pairs of random servers. The results are identical for \sysname{} and fat tree. \newtext{Any differences in path latencies of the two networks are immaterial for application performance.}

\autoref{fig:prod_combined} (left) compares the packets per second (PPS) for small, 64-byte packets between pairs of servers. \sysname{} is slightly higher because it has less background traffic.
\autoref{fig:prod_combined} (right) shows I/O operations per second (IOPS) for storage reads from clients to multiple servers. \sysname{}'s performance matches fat trees for this workload as well.

These benchmarks, and the absence of user-reported performance issues, confirm that \sysname{} matches fat tree performance (at lower cost and higher fault tolerance).

\section{Broader Evaluation}
\label{sec:evaluation}

Production fabric benchmarks can validate the performance of real applications but we need simulations to evaluate a broad range of design parameters and workloads. 
We benchmark \sysname's oversubscription ratio and the number of edge disjoint paths that Spraypoint produces. We also compare its throughput and cost relative to fat trees. 

\subsection{Oversubscription}
\label{sec:eval-oversub}

\begin{figure*}[t!]
    \centering
    \begin{subfigure}[b]{0.262\textwidth}       
        \includegraphics[width=\columnwidth]{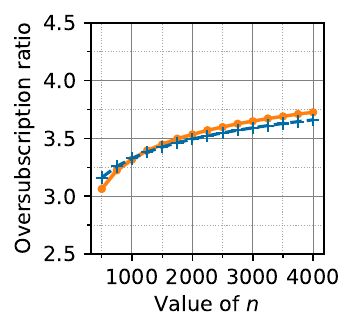}
    \end{subfigure}
    \hfill
    \begin{subfigure}[b]{0.23\textwidth}       
        \includegraphics[width=\columnwidth]{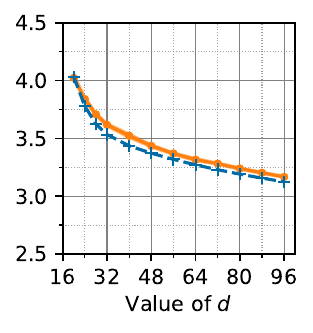}
    \end{subfigure}
    \hfill
    \begin{subfigure}[b]{0.23\textwidth}       
        \includegraphics[width=\columnwidth]{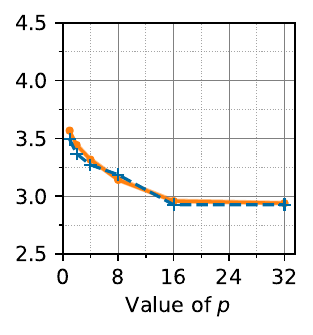}
    \end{subfigure}
    \hfill
    \begin{subfigure}[b]{0.23\textwidth}       
        \includegraphics[width=\columnwidth]{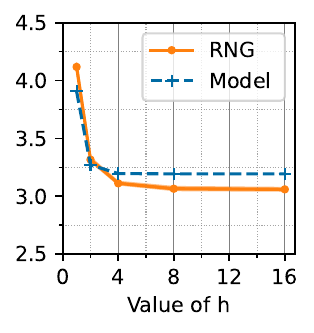}
    \end{subfigure}
    \hfill
    \caption{Oversubscription for different topology parameters. The defaults are $n$=1000, $d$=64, $p$=4, $h$=2.}
    \label{fig:oversub}
\end{figure*}

Oversubscription ratio characterizes worst-case throughput of a network and helps quantify congestion risk. 
Namyar et al.~\cite{throughput-centric2021} prove that worst-case throughput occurs for a perfect unidirectional matching, where each node sends at full rate to its counterpart.
The authors do not suggest a way to generate the worst matching among the many choices. We randomly generate 100 matchings and estimate oversubscription using the worst value, though we will see that the distribution is narrow because the topology and routing are random~\cite{FrMe12,MoRa-book}.
Similar to our oversubscription model, this process characterizes stochastic oversubscription.

For a given matching, the oversubscription ratio is $r$ if each transmitter can send at least $1/r$ fraction of its full rate without congesting any link. This measure is determined by the smallest blocked flow, irrespective of whether others can send more.  We compute $r$ by solving a linear program (LP) that encodes a multi-commodity flow problem~\cite{wikipedia_mcf}. Traffic for each flow is spread across Spraypoint paths, and the LP minimizes $r$ under link capacity constraints. These LPs are huge and each takes multiple hours to solve, which limits the largest fabrics we can practically study.

\autoref{fig:oversub} shows the results of simulations ("\sysname{}") and our model (\autoref{sec:model:oversub}). The line denotes the minimum value of $r$ for each setting. There is a shadow for the max but it is barely visible because the variance is near zero.
To study a broad range of settings, we vary one parameter while the others are set to default values ($n$=1000, $d$=64, $p$=4, $h$=2). This default can support 64K servers with 100 Gbps uplinks and offers an oversubscription ratio of 3.25. 
We see that the oversubscription ratio changes gracefully as we vary the parameters, highlighting the ability of random graphs to support fine-grained oversubscription levels. We also see that the impact of $h$ and $p$ flattens out beyond a point. 

The model matches empirical results well, enabling performance predictability. It allows operators to plan \sysname{} fabrics like they plan fat trees today, with the added advantage of lower cost, higher fault tolerance, and finer-grained performance tuning. 
Some settings in \autoref{fig:oversub} are technically outside the modeled regime, e.g., $d = 20 \not\geq 2(\ln(n=1K) + 5)$. 
We have defined the regime conservatively to get high probability results, and the models degrade gracefully outside the regime.

\subsection{Edge disjoint paths}
\label{sec:eval-spraypoint}

Spraypoint aims to find many edge disjoint paths. 
We evaluate this ability by computing the min cut across paths between endpoint pairs, which equals the number of edge disjoint paths~\cite{wikipedia_mengers_theorem}.
\autoref{fig:eval-spraypoint} shows the CDF of min cut across endpoint pairs for the same default fabric parameters as above ($n$=1000, $d$=64, $p$=4, $h$=2). For comparison, it also shows $k$-shortest-path routing (which is not implementable at scale with commodity switches). 
We consider $k$=8~\cite{jellyfish2012,xpander2016} and $k$=64, a high value that equals the node degree (and needs 8x more resources). 

For almost all endpoint pairs, Spraypoint finds over 50 edge disjoint paths. For half the pairs, it finds over 60 such paths out of the maximum possible 64. 
In contrast, the median for 8-shortest-paths is 5 and for 64-shortest paths is 35. This difference impacts oversubscription directly. The oversubscription for the two k-shortest-paths configurations are 21.3 and 4.7, while that for Spraypoint is 3.25.
Spraypoint is easier to implement \emph{and} finds a greater diversity of paths. It offers higher peak throughput between endpoint pairs, and offers higher network-wide throughput.

\begin{figure}
    \centering
    \includegraphics[width=\columnwidth]{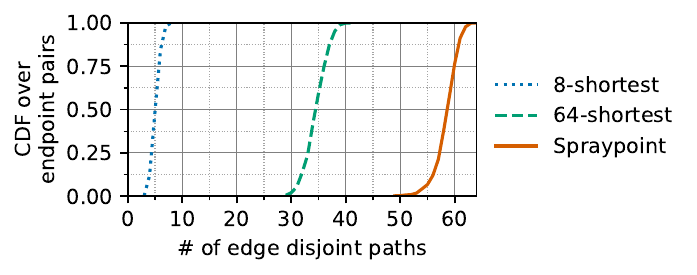}
    \caption{Edge disjoint path count.}
    \label{fig:eval-spraypoint}
\end{figure}

\begin{figure*}[t!]
    \centering
    \includegraphics[width=\textwidth]{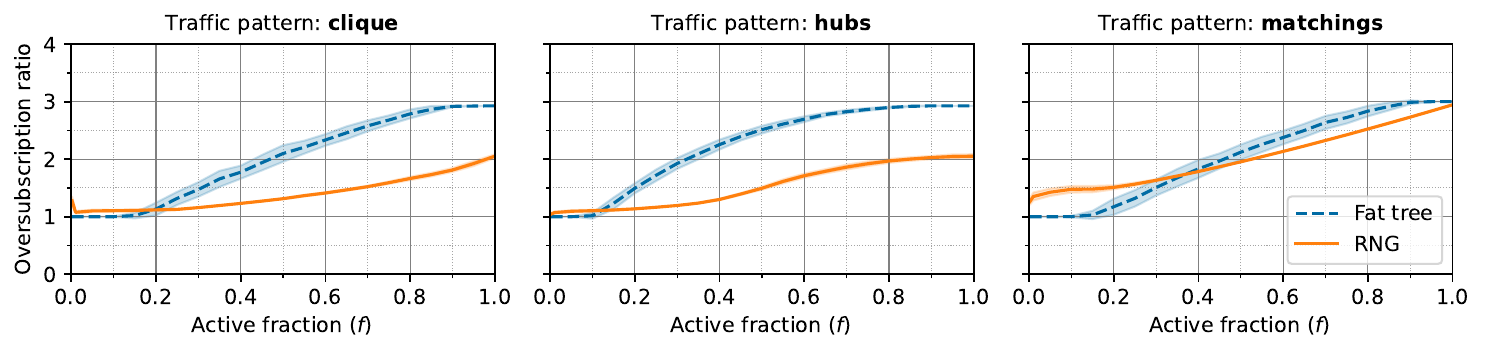}
    \caption{Oversubscription ratio (lower is better) for different traffic matrices. Both fat tree and \sysname topologies are configured for worst-case oversubscription ratio of 3:1. \sysname topologies use 45\% fewer switches.}
    \label{fig:throughput-patterns}
\end{figure*}

\subsection{Throughput relative to fat trees}
\label{sec:eval-throughput}

Beyond worst-case traffic patterns that help characterize oversubscription, operators want to also validate performance for other traffic patterns. This validation is hard because there are many possible traffic patterns.
Our evaluation is based on the observation that any traffic pattern may be viewed as a combination of three basic types, and we can get insight into the performance of \sysname{} for a range of traffic patterns by studying these types.
The three types are:
\begin{asparaitem}

\item {\em Clique:} A group of nodes exchanges traffic amongst themselves (e.g., collective communication like all-reduce, replication traffic among storage servers).  

\item {\em Hubs:} Some nodes (e.g., Web or storage servers) send to and receive from all other nodes. 

\item {\em Matchings:} Nodes send all traffic to exactly one other node, as we studied earlier.
    
\end{asparaitem}

The nodes in the traffic patterns are ToRs, not individual servers, because that stresses the fabric more~\cite{throughput-centric2021}. 

Each pattern type is parameterized by active fraction $f \in [0, 1]$ to capture the skew. 
Clique($f$=0.2) means that the clique size is 20\% of the nodes (randomly selected); hubs($f$=0.2) means that 20\% of the nodes in the network are hubs; and matchings($f$=0.2) means that 20\% of the nodes are involved in the matching.  
For each pattern type and $f$, all flows in the traffic matrix (between a sender-receiver pair) are equal and sum of rows and columns equal the local capacity of nodes. Thus, each matrix is routable in a non-blocking fabric. 

For each value of $f$, we generated 100 random traffic matrices. We quantify performance for a matrix using the fraction $r$ as before---the network can carry at least $1/r$ fraction of each flow. Thus, we are computing the oversubscription ratio of each matrix, instead of worst-case oversubscription ratio.

\autoref{fig:throughput-patterns} shows the results for \sysname{} and fat tree. The lines denote the mean and the shaded area denotes the min and max. 
Both topologies are designed to support 61.4K servers across 960 ToRs, with a worst-case oversubscription of 3:1. In this configuration, \sysname{} uses 45\% fewer switches.

We see that \sysname{} performs better by as much as 30\% across a wide regime for clique and hubs. The exceptions are cases with $f$<0.1, where fat tree is 5-10\% better. For these cases, shortest path routing on fat trees finds edge disjoint paths that equal the ToR degree. The equivalent number is 10\% lower in \sysname{} for $h$=2 (see \autoref{sec:eval-spraypoint}). The number of edge disjoint paths matters most when there are few senders in the network. For matchings, fat tree is better when $f$<0.4 and \sysname is better otherwise.

\cloudone operators prefer the broadly better performance of \sysname{}, especially given its lower cost and higher fault tolerance. The regimes with low values of $f$, where fat tree performs better, are acceptable because they will materialize only when all servers in the rack act in unison, a rarity for multi-tenant datacenters. 

\subsection{Cost relative to fat trees}
\label{sec:eval-cost}

Prior works report a range of values on the cost of expander topologies relative to fat trees. These values are difficult to compare because they use different performance measures and network sizes. It is also difficult to get a holistic view because of their experimental methodology. 
Our models enable systematic analysis of the relative cost as a function of oversubscription ratios and network sizes. 

We use the ratio of the number of switches used by the two topologies as the measure of relative cost. This measure automatically includes transceivers, whose count is proportional to switch count. It ignores passive optical components like cables, connectors, and patch panels; the cost of these components is in the minority and varies based on specific cabling methods (e.g., if cables are trunked) in the datacenter. Switch count is a general, reproducible measure. 

\begin{figure}[t!]
    \centering
    \includegraphics[width=\columnwidth]{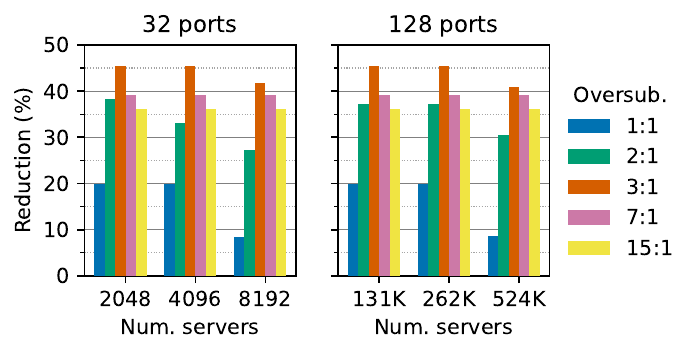}
    \caption{Reduction in number of switches in \sysname{} relative to a 3-tier fat tree for two different port counts.}
    \label{fig:cost}
\end{figure}

We compute the number of switches needed for \sysname{} using the procedure in \autoref{sec:fabric-planning} and compare it to generalized 3-tier fat tree with the same oversubscription level. ToRs are not oversubscribed in either \sysname or fat trees. 
\autoref{fig:cost} shows the results for different oversubscription levels and two possible values of switch port counts. The highest server count for each port count is the maximum supported by a non-blocking fat tree~\cite{taming2011}, and the other two are 50\% and 25\% of the servers. Oversubscription of 2:1 cannot be precisely supported in fat trees with 32 or 128 port switches, but we assume each port is 3-way splittable for this experiment. 

The cost reduction varies by 5x, between 9 and 45\%. It has a similar pattern across port counts and network sizes, and oversubscription ratio is the key determiner. 
The reduction is low for the 1:1 oversub case because non-blocking fat trees do not strand capacity; the reduction here primarily comes from shorter paths in \sysname. After increasing from oversubscription of 1:1 to 3:1, the cost advantage starts decreasing slowly. For high oversubscription, a fat tree has fewer switches at the aggregation and spine tiers, which lowers the potential of reducing cost by removing those tiers. 

Our analysis makes it clear when expander topologies like \sysname{} bring the most or least value, allowing operators to make informed choices about the type of topology to use. 
\section{Related Work}
\label{sec:related}

We build on the foundations laid by a long line of prior research. 
\change{The notion of an expander, or an "optimal" graph for routing purposes, was defined in the early 90s~\cite{HoLiWi07}, and it has been long known that random graphs are nearly-optimal expanders~\cite{Fr08}).} Several researchers have proposed expander topologies for datacenters and provided insights into their performance~\cite{jellyfish2012,xpander2016,slimfly2014,beyond-fat-trees2017,throughput-centric2021,jyothi2016measuring}.
Particularly influential works for us include Jellyfish~\cite{jellyfish2012}, which proposed (truly) random graphs for datacenters; Xpander~\cite{xpander2016}, which made the connection that random graphs do well for data center workloads because they are expanders; and Namyar et al.~\cite{throughput-centric2021}, who outlined throughput bounds of expander topologies.  
We extend this body of work by addressing unsolved challenges related to routing, cabling, and performance predictability. We also deploy the first expander-based networks in production.

Researchers have also proposed non-expander topologies that tackle some of the challenges of the fat trees (e.g., cost), such as HyperX~\cite{hyperx2008}, DCell~\cite{dcell2008}, and BCube~\cite{bcube2009}. We choose random graphs because they have significantly shorter path lengths.
Further, cabling these other topologies for large, multi-room datacenters is still an open challenge. 

Capacity fungibility, a key goal of the \sysname{} design, can be achieved using reconfigurable hardware as well. We use expander topologies because they need neither a (logically) centralized control plane nor non-standard hardware. Most designs with reconfigurable hardware require a control plane that predicts global traffic demands and dynamically reconfigures the topology~\cite{c-through2010,firefly2014,flyways2011,projector2016,helios2010}. The scale of our datacenters, presence of large amounts of bursty traffic (common for Web workloads), and reconfiguration delays complicate the development and operation of such control planes. 

There are "demand-oblivious" designs that do not require a centralized control plane~\cite{rotornet2024,Sirius2020,opera2020}. They use novel optical devices that rapidly cycle through a fixed schedule. All-to-all connectivity is not available at any given time, and the systems use forwarding via intermediate nodes based on Birkhoff-von Neumann traffic matrix decomposition~\cite{chang2002load}. In addition to relying on non-commodity hardware, these designs bring a host of software risks such as time synchronization, packet reordering within a TCP flow, some throughput degradation~\cite{rotornet2024}, and switch-level congestion control~\cite{Sirius2020}.

\section{Conclusions}
\label{sec:conclusions}

Flat expander topologies based on \change{quasi-}random graphs can be practically realized using the routing and cabling approaches we developed. They can be designed for the desired level of performance and cost using the models we developed.

\smallskip
\noindent

\bibliographystyle{ACM-Reference-Format}
\bibliography{reference}

\begin{appendix}
\section{Incremental cabling}
\label{sec:app:incremental}

The physical connectivity of \sysname described in \autoref{sec:cabling} can be achieved incrementally.
While the first room of the datacenter is being prepared, we deploy its patching panel. The number of ShuffleBoxes in the panel is based on the number of ToR uplinks expected in the room. Initially, all r-ports and c-ports have ShuffleBacks in them. When racks land in the room, ToRs are connected to random r-ports, as part of which the ShuffleBacks of selected r-ports are removed. The ToRs start forming a random graph via c-port ShuffleBacks, as in \autoref{fig:connectivity-patterns}(a). Not all ToR uplinks may find a match during this process (more on this below).

While the second room is being prepared, we land its shuffle panel and connect those c-ports to the panel in the first room. If panels are equal-sized, half of the c-ports in each must connect to the other panel; otherwise, we adjust proportionally. The ports are picked randomly. The c-port connections are enabled by the trunk cables between the two rooms.  Making these connections requires removing ShuffleBacks from c-ports, as a result of which the connectivity pattern for some ToRs may go from \autoref{fig:connectivity-patterns}(a) to \autoref{fig:connectivity-patterns}(c), except that both ToRs are in the same room.

When racks can land in the second room, those ToRs connect to random r-ports in the second room's panel. Some of these ToR uplinks will connect to ToRs in the first room via c-port connections, as in \autoref{fig:connectivity-patterns}(b); some will connect to other ToRs in the same room via c-port ShuffleBacks; and, at least initially, some will fail to find a match. 

When the third room starts, we land its panel and rebalance c-port connectivity. When there were two rooms, half of the first room's c-ports connected to the second room. When the third room starts, a third must go to the second room and a third to the third room (and a third have ShuffleBacks). To rebalance c-port connectivity, we break some existing c-port connections (selected randomly) between the first two rooms and use the freed up c-ports to connect to the third room. 
We follow the same process of landing panels and rebalancing connectivity for all subsequent rooms. 

Rebalancing connectivity breaks logical links that carry live traffic. We thus drain impacted links to minimize impact. Strictly speaking, logical connectivity changes when racks land and r-port ShuffleBacks are removed, but the blast radius of that operation is much smaller ($f_r$=4 versus $f_c$=32).

\subsection{Growing pains} 

An artifact of our incremental construction is that early in the growth of the datacenter, not all ToR uplinks form valid adjacencies. For an uplink in room 1 to form an adjacency, it must connect to an r-port FP that is bridged (by the ShuffleBox's c-port ShuffleBack) to an r-port connected to another uplink. When only a small fraction of r-ports are populated, this probability is low. 
A similar effect occurs for the second room but is less severe because half of ToR uplinks connect to r-port FPs that are connected to the first room via c-ports and find matches there.

\begin{figure}
    \centering
    \includegraphics[width=\columnwidth]{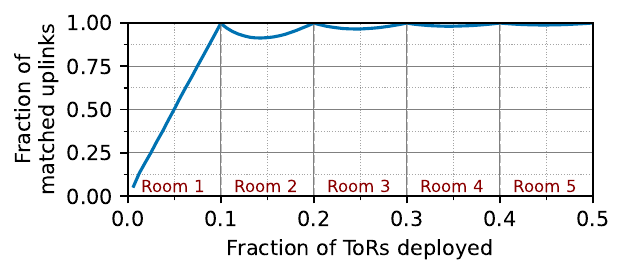}
    \caption{Matched uplinks as the datacenter grows.}
    \label{fig:incremental-fp}
\end{figure}

\autoref{fig:incremental-fp} plots the fraction of deployed uplinks with valid adjacencies as we deploy more ToRs. This simulated datacenter has 10 rooms with 100 ToRs each.  We see the fraction of uplinks in room 1 that find a match increases linearly.
After room 2 opens, there is an initial dip in the fraction of matched uplinks because room 2 ToRs have fewer matched uplinks and their ratio in the population is growing. But the lowest point in the dip is still above 90\%. The impact is negligible beyond room 2. 

We present a model of datacenter expansion that accurately predicts the average degree at each point in the datacenter growth. Operators can use the model to predict the average degree during expansion and plan accordingly.

\begin{figure}[t!]
    \centering
    \includegraphics[width=\columnwidth]{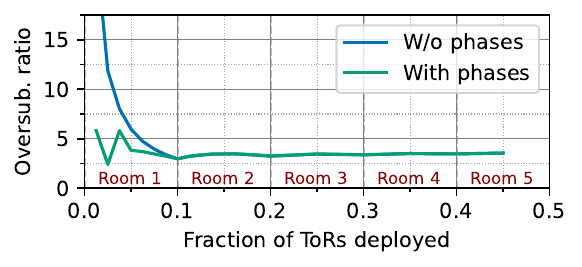}
        \caption{Oversubscription as the datacenter grows.}
        \label{fig:incremental-oversub}
\end{figure}

As the "W/o phases" curve of \autoref{fig:incremental-oversub} shows, these growing pains lead to poor early throughput for room 1 (and room 1 only). 
\cloudone operators view this behavior as acceptable for fast-growing datacenters. The window of poor performance is short and they can delay on-boarding applications. 

\label{sec:appendix:phases}

But the poor early performance can be problematic for slow-growing datacenters. For these cases, we propose a modified cabling scheme. 
We partition the ShuffleBoxes of the first room's shuffle panel into two or more phases. Initial ToRs connect only to phase 1 ShuffleBoxes. When these ShuffleBoxes are full, we connect the c-ports of phase 1 and phase 2 ShuffleBoxes and start mapping future ToRs to phase 2 ShuffleBoxes. This process is repeated for future phases. It essentially divides a room into smaller, logical rooms and shrinks the poor performance window. 

The "With phases" curve of \autoref{fig:incremental-oversub} shows the impact of phased cabling. We partitioned room 1 into two phases, and the size of the first phase was 30\% of the room. 
This configuration is optimal when operators want average degree of the datacenter to always stay above 0.8 as soon as 25\% of the racks in room 1 are deployed. 
Next, we show how to derive optimal phase sizes for user-defined performance criteria. 

\subsection{Modeling expansion} 
\label{sec:appendix:incremental}

We present a model of datacenter expansion that predicts the average degree of the fabric at each point in time. The expansion occurs in \emph{stages}, where each stage corresponds to the landing of new panels of ShuffleBoxes and associated re-cabling. Routers of a stage only connect to ShuffleBoxes of that stage. A stage could be a room or a phase within a room, and the analysis makes no assumption about the number or relative sizes of various stages.

\label{sec:appendix:average-degree}

We model datacenter expansion using ``time" $t$ as the primary variable, where $t \in [0,1]$ is the fraction (relative to the full DC) of routers currently landed. We denote a stage that begins at $t_1$ and ends at $t_2$ as $[t_1, t_2]$.
Abusing notation, we will simply denote a router by its timestamp of arrival. 
Each router has $d$ uplinks.

During the first stage, when there are no routers from previous stages, the probability of a ToR uplink finding a match increases linearly because the probability depends on bridging to an occupied r-port which grows linearly with $t$. The average degree of the fabric at time $t$ is d$\frac{t}{T_0}$, where $T_0$ is the size of the first stage. 

For subsequent stages, the following invariants hold about \sysname's datacenter expansion:
\begin{enumerate}
    \item At the end of any stage, the graph degree is $d$ because all routers would have found a match. 
    
    \item During a stage, all routers from previous stages maintain their degree of $d$.

    \item When a router on a stage lands, only some of its uplinks find a match. On expectation,
    that fraction is the fraction of routers present. Formally, for stage $[t_1, t_2]$, there are $t_1$ routers from previous stages, and we will land routers $t_1+1, t_1+2, \ldots, t_2$. The $i$-th router in this stage will make (on expectation) a $\frac{t_1 + i}{t_2}$ fraction of connections when it lands. 
\end{enumerate}

Based on these invariants, we can prove:

\begin{model} \label{mod:incremental} 
The average degree during the stage $[t_1, t_2]$ ($t_1 > 0$)  is $d \Big[ \frac{t_1}{t} + \frac{t-t_1}{t_2}\Big]$
\end{model}

\begin{proof} The stage $[t_1, t_2]$ begins with a $t_1$ fraction of routers connected (with degree $d$) and ends with a $t_2$ fraction of routers. 
At time $t \in [t_1, t_2]$, all routers with timestamp at most $t_1$ always have a degree of $d$, and the routers of the current stage will have, on average, a degree of $dt/t_2$.
The fraction of routers with timestamp at most $t_1$ is $t_1/t$, and the fraction of routers of the current stage is $(t-t_1)/t$.

Thus, the average degree at time $t$ is:
\begin{equation}
    \frac{t_1}{t} \cdot d + \frac{t-t_1}{t} \cdot \frac{dt}{t_2} = d \Big[ \frac{t_1}{t} + \frac{t-t_1}{t_2}\Big]
\end{equation}
\end{proof}

\begin{figure}[t!]
    \centering
    \includegraphics[width=\columnwidth]{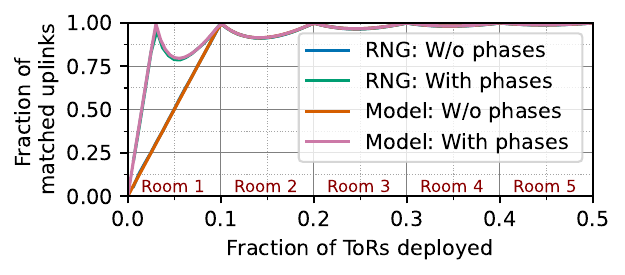}
    \caption{Comparison of fraction of uplinks that find a match based on simulation of \autoref{fig:incremental-fp} versus what is predicted by the model.}
    \label{fig:incremental-fp-model}
\end{figure}

\autoref{fig:incremental-fp-model} compares \Mod{incremental} with a simulated fabric from \autoref{fig:incremental-oversub}. 
We see that the model matches simulations well, and the curves are virtually indistinguishable. 

\subsection{Computing optimal phases} 
\label{sec:appendix:optimal-phases}
The expansion model enables us to find optimal phases. We formulate the problem as: Find the minimum number of phases and their sizes such that the average degree is at least $\alpha d$ when at least a $\beta$-fraction of the first room is deployed. The parameters $\alpha$ and $\beta$ are operator-provided inputs. 

We solve the problem by applying \Mod{incremental} to the first room, treating each phase as a stage. We know that the average degree increases linearly in the first stage. So, if $\alpha \leq \beta$, we only require a single phase (i.e., not using phases at all). If $\alpha > \beta$, our aim is to maximize $\alpha$ for a given $\beta$, or equivalently, minimize $\beta$ for a given
choice of $\alpha$. 

Let us now consider the case of two phases in a room. 

\begin{theorem} \label{thm:two-stage} With two phases in a room, for a given $\alpha$, the minimum possible $\beta$ is $ \alpha (1-\sqrt{1-\alpha})^2$.
Further, the minimum occurs when the size of the first phase is $\beta/\alpha$.
\end{theorem}

\begin{proof} Let us first observe that the formula in \Mod{incremental} is a hyperbola ($t_1/t + t/t_2$). 
Its minimum is at the geometric mean $\sqrt{t_1 t_2}$, so the minimum average degree is:
\begin{equation}
\label{eq:min-degree}
    d(2\sqrt{t_1/t_2} - t_1/t_2)
\end{equation}

Let us now denote the minimum possible $\beta$ as $\beta^*$. In the first phase, the average degree increase linearly. At time $\beta^*$, it needs
to be at least $\alpha d$. This phase ends when the average degree is $d$, which happens
when time is $\beta^*/\alpha$, which we denote using $x$ for convenience. At this point, the next phase begins and continues to the end of the room.

If the end of room deployment as $t=1$, we have the stage being $[x, 1]$.
Applying \autoref{eq:min-degree}, the minimum average degree
is $d(2\sqrt{x} - x$). 
Since we require the average degree to be $\alpha d$ at all times, $2\sqrt{x} - x \geq \alpha$.

Consider the quadratic $f(x) = x - 2\sqrt{x} + \alpha$. 
We need to choose $x$ so that $f(x) \leq 0$,
which happens after the first root of the quadratic.
The roots of $f(x)$ are $$ \frac{2 \pm \sqrt{4 - 4\alpha}}{2} = 1 \pm \sqrt{1-\alpha} $$

The smaller root is $1-\sqrt{1-\alpha}$, and we need $x$ to be at least this value.
Plugging in the expression for $x$,
$$ \sqrt{\frac{\beta^*}{\alpha}} \geq 1 - \sqrt{1-\alpha} \ \ \ \ \Longrightarrow \ \ \ \ \beta^* \geq \alpha(1-\sqrt{1-\alpha})^2 $$
\end{proof}

\begin{figure}[t!]
\centering
 \includegraphics[width=\columnwidth]{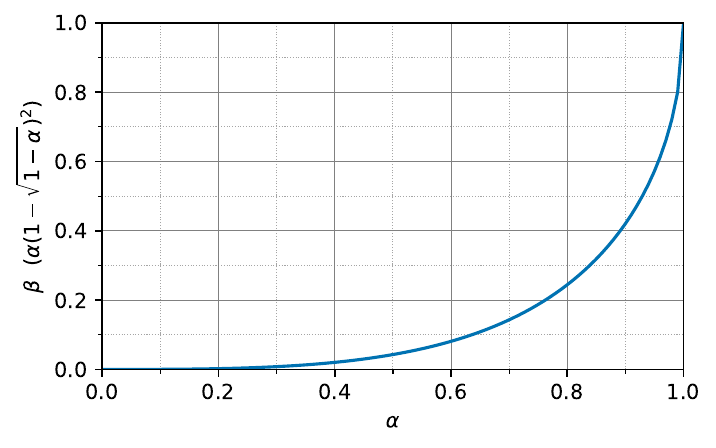}
\caption{The minimum $\beta$ achievable, for a given $\alpha$ in a two-phase deployment} \label{fig:alpha-beta}
\end{figure}

For a better understanding of the formula in \Thm{two-stage}, \autoref{fig:alpha-beta} plots the optimal
$\beta$ versus $\alpha$. We see if we want the degree to be at least $0.8d$, we can get that as soon as $\beta\approx25\%$ of the room is deployed. This performance can be obtained if the first phase is $\approx$30\% of the first room. 

If the $\beta^*$ of the two-phase deployment does not meet the operator criterion, we can perform a similar analysis for successively higher number of phases until we find the number of phases where the criterion is met. These iterations are guaranteed to terminate. In the extreme, each router is its own phase and the graph degree is always $d$.

\section{Latency Reduction}
\label{sec:latency}

While expander topologies have fewer hops than fat trees, we find that when deployed in a large datacenter their latency can be worse because hops can be long (half the datacenter length on average). The higher propagation delay exceeds what is saved in switching costs due to fewer hops.

\begin{figure}[t!]
    \centering
    \includegraphics[width=\columnwidth]{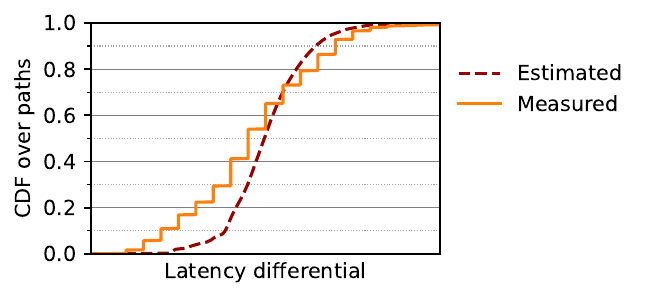}
    \caption{Comparison between estimated latency based on datacenter layout and measurements on the production fabric. Units omitted for confidentiality.}
    \label{fig:prod_latency}
\end{figure}

\begin{figure}[t!]
    \centering
    \includegraphics[width=\columnwidth]{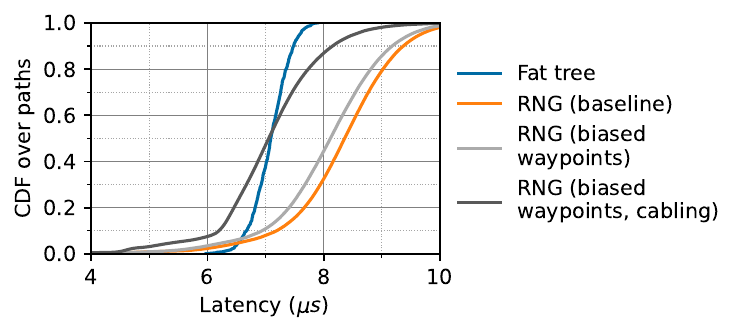}
    \caption{Latency distribution of ToR-to-ToR paths}
    \label{fig:latency}
\end{figure}

To demonstrate this challenge, we simulate a datacenter with a span of 300 meters and 4K ToRs spread across 10 equally-sized rooms. For \sysname{}, the patching panels are located in the center of each room. The fat tree has 4 tiers. 
ToRs connect to aggregation pods with two tiers of switches, where the bottom tier connects to ToRs and the top tier to spine routers. Aggregation pods are located in the center of each room to serve the ToRs in the room, and all spine routers are in one room in the middle of the datacenter~\cite{facebook-dc}. For both topologies, the link speed is 100 Gbps 
and cables travel along cable trays~\cite{cable-trays}. \autoref{fig:prod_latency} validates the accuracy of our latency estimation methodology by comparing what it estimates for the production server mesh (\autoref{sec:implementation}) with measurements on the fabric.

While the fat tree has 6 hops and \sysname{} has 4 hops for the vast majority of ToR pairs in the simulated datacenter, \autoref{fig:latency} shows that the median one-way latency of ("\sysname{} (baseline)") is 15\% higher (8.4 vs. 7.1 $\mu s$). This difference may seem small, but it compounds over round trips for applications and so datacenter operators are sensitive to it. 

Simple modifications to routing and cabling reduce the latency of \sysname{} without compromising performance.

\parab{Biased waypoint selection} When selecting waypoints, instead of picking $p$ random neighbors of $v \in Nbr(t)$, prefer waypoints closer to $t$. We implement this preference by ranking each candidate waypoint $w$ using $room\_dist(w, v) + room\_dist(v, t)$, where $room\_dist()$ is the distance between rooms of the two nodes. Using coarse, room-level distances, instead of actual fiber runs, avoids systematically preferring advantageously-located waypoints (e.g., close to cable trays). 

\parab{Biased cabling} We lower the number of inter-room connections, which are longer than intra-room ones. 
Specifically, a patching panel makes $\alpha \times \frac{1}{R}  \times C \times d_c$ c-port connections to another panel, where $C$, $R$, $d_c$ are the number of ShuffleBoxes, rooms, and c-ports per ShuffleBoxes, and $\alpha \in [0, 1]$ is a parameter that controls the amount of inter-room connections. For the baseline construction, $\alpha=1$. We use $\alpha=0.5$, which reduces number of inter-room connections by half. It reduces latency without hurting throughput because there is still enough inter-room capacity. Another positive side-effect of this optimization is that it reduces (by $1 - \alpha$) the number of inter-room cables, an important factor behind cabling costs. 

\autoref{fig:latency} shows that the combination of these modifications makes the median latency of \sysname{} similar to that of the fat tree. Biased waypoint selection reduces the median latency by 0.3 $\mu s$ and biased cabling by a further 1 $\mu s$. 

Latency can be further reduced by limiting spraying to a subset of neighbors that have shorter paths to the destination and by selecting the $h$ next hops in a distance-aware manner. But, unlike the two modifications above, these modifications reduce, respectively, the number of edge disjoint paths and throughput for some traffic matrices. They must thus be used only in settings where reduced performance is acceptable. 

We experimented with a few other ideas, such as spraying to only a subset of close neighbors and biasing the selection of $h$ random next hops, but they either reduced throughput or did not bring substantial benefit.

\section{Modeling Overview}
\label{sec:model:overview}

Our modeling of \sysname prioritizes simpler formulas over mathematical precision. It makes numerous modeling assumptions and performs asymptotic analysis (as $n$ gets large) that ignores lower order terms in some situations. 
We constrain the operating regime to where the inaccuracy of these assumptions and asymptotics is minimal.
Recall from \Sec{modeling} that this regime is (i) $2(\ln(n) + 5) \leq d \ll n$, (ii) $p \geq (n/d^2)^{1/\ell}$,
and (iii) $h \ll d$, where $n$ and $d$ are the node count and degree. 
Further, given random topology and routing, the analysis is inherently probabilistic, i.e., results hold with high probability and can fail with low probability.

The following sections derive the performance models in \Sec{modeling} in a stepwise manner.

{\bf Random graph preliminaries (\Sec{prelim}):} The first step of the modeling is to define a probabilistic
process that generates the random network topology. We give a precise formulation,
since this forms the foundation of the analysis. We prove some basic theorems
about the graph construction. This section is purely using random graph theory,
and does not involve any asymptotics or modeling of Spraypoint/routing.

{\bf Modeling path length (\Sec{level}):} We begin the actual analysis by
giving formulas for the sizes of various Spraypoint levels. 
The random graph theorems from
\Sec{prelim} are used to determine the structure of the Spraypoint levels
and their sizes. The path length
distributions follow quite directly from these formulas. These formulas use
asymptotics to get convenient expressions. This yields Model~\ref{mod:path-length}.

{\bf Modeling edge disjoint paths (\Sec{mincut}):} Armed with the tools from 
the previous sections, we can analyze the mincut properties between
a source-destination pair. For this analysis, we make some minor modeling
assumptions, to represent the flow as a tractable balls and bins problem.
A detailed, tight analysis of these problems has been done before~\cite{FrMe12},
but we observe that a simpler upper bound formula is close enough.
This bound is given in Model~\ref{mod:edge-disj}.

{\bf Modeling oversubscription (\Sec{app:oversub}):} The most difficult 
part of our analysis is bounding the oversubscription ratio. For this
analysis, we use many modeling assumptions (detailed in \Sec{oversub-principles}).
The oversubscription ratio is the output of a multicommodity flow
linear program, so it is much harder to get a simple formula. Moreover,
the mathematics of interacting flows is complex and we resort to various
simplifications based on intuitions on random graphs. Our final
formula is given by a collection of polynomial equations in the various
parameters, as given in \Sec{model:oversub}.

\section{Random graph Preliminaries} \label{sec:prelim}

The router graph $G = (V, E)$ has $n$ nodes and degree $d$. We model its construction as a random configuration graph (Chap. 2 of~\cite{Wo-book}). 
There are various ways to construct graphs in the configuration model, all of which are statistically equivalent and provide convenient analytical tools. 

\medskip

\begin{asparaitem}
    \item The standard method: consider each node as incident to $d$ ``half-edges". Paired
    half-edges result in a proper edge. We generate a pairing with a uniform random permutation of all half-edges,
    and matching the first to the second, the third to the fourth, and so on. 
    \item Getting the neighborhood of a set $S$: Condition on some edges within $S$. Each remaining half-edge
    in $S$ picks \emph{independently} another uar (uniformly at random) half-edge in graph to pair with. If multiple $S$ half-edges
    pick the same partner to pair with, only one them (picked randomly) will pair up. The others remain
    unpaired. At this stage, some of the half-edges in $S$ have paired up. Now, we simply pair up all remaining
    half-edges in the graph using the standard method (which has dependencies). 
    The first part makes a collection of independent decisions, allowing for easier analysis. 
\end{asparaitem}

\subsection{Probability preliminaries} \label{sec:prob}

We list some probability theory definitions and theorems used in the analysis.
In what follows, capital letters $X, Y, Z$ denote non-negative integer-valued random variables; $\EX[X]$ denotes the expectation of $X$. script letters $\cE, \cF, \ldots$ denote events; and $\Pr[\cE]$ denotes the probability of event $\cE$.

The first theorem will help us derive expected values of a collection of random variables. It makes no independence assumption on the random variables. 

\begin{theorem} \label{thm:lin} [Linearity of Expectation] For any finite collection
of random variables $X_1, X_2, \ldots, X_k$, $\EX[\sum_{i \leq k} X_i] $ $= \sum_{i \leq k} \EX[X_i]$.
\end{theorem}

The next theorem is convenient for dealing with ``error bounds" and bounding the probability of bad events.

\begin{theorem} \label{thm:ub} [Union bound] Given a finite collection of events $\cE_1, \cE_2, \ldots, \cE_k$,
$\Pr[\bigcup_{i \leq k} \cE_i] \leq \sum_{i \leq k} \Pr[\cE_i]$.
\end{theorem}

The following states classic \emph{concentration inequalities} that upper bound the probability of a random variable deviating significantly from its expectation.  

\begin{theorem} \label{thm:chernoff} [The Chernoff bound] (Theorem 1.1 of~\cite{DuPa-book})
Let $X = \sum_{i \leq k} X_i$, where the $X_i$'s are independent random variables in $[0,1]$.
Then,
\begin{asparaitem}
    \item For all $\eps \in (0,1)$, $\Pr[X < (1-\eps) \EX[X]]$, $\Pr[X > (1+\eps) \EX[X]]$
    $\leq \exp(-\eps^2 \EX[X]/3)$.
    \item For all $t$, $\Pr[X > \EX[X] + t]$, $\Pr[X < \EX[X] - t]$ $\leq \exp(-2t^2/k)$.
    \item If $t > 2e\EX[X]$, then $\Pr[X > t] \leq 2^{-t}$.
\end{asparaitem}
\end{theorem}

We use the phrase \emph{with high probability} to denote probability at least $1-o(1)$, with asymptotics over increasing $n$.

\subsection{Random graph tools} \label{sec:rand-prop}

We state and prove some probabilistic graph theorems that will be useful in the analysis.

\begin{claim} \label{clm:nbr} Consider a node set $S$. Condition on any choice of edges
within $S$ that leave $k$ half-edges unpaired. The probability that a vertex $v \notin S$ is a neighbor of $S$ is at least $1-\exp(-k/n)$.
\end{claim}

\begin{proof} Conditioned on the current choice of edges in $S$, let us consider the process of  building the neighbors of $S$. As described above, in the first step of this
process, each of the $k$ half-edges incident to $S$ picks a uar unpaired half-edge to pair with. Fix $v \notin S$. 
Let us look at the probability that a single iteration connects with one of the $d$ half-edges incident to $v$.
This probability is at least $d/nd = 1/n$. 

Since there are $k$ independent iterations, one for each unpaired half-edge
incident to $S$, the probability that no iteration connects with $v$ is at most $(1-1/n)^{k}$. Applying
the inequality $1-x \leq \exp(-x)$, we upper bound by $\exp(-k/n)$. Thus, the probability
of $v$ connecting to $S$ is at least $1-\exp(-k/n)$.
\end{proof}

Let us unpack this expression.
Using the approximation (for small $x$), $\exp(-x) \approx 1-x$,
roughly speaking, we can approximate as $1-(1-k/n) = k/n$. 

The following claim shows that all sets of sufficiently large size connect with every other node. It is a special case of the proof that the configuration model generates expanders.

\begin{claim} \label{clm:straggler} Fix subset $S \subseteq V$ and condition on any choice of edges within $S$ that leaves at least $nd/4$ unpaired half-edges incident to $S$.
With probability at least $1-n^{-4}$, every $v \in V \setminus S$ has an edge to $S$.
\end{claim}

\begin{proof} By \Clm{nbr}, the probability that $v \notin S$ is \emph{not} a neighbor of $S$
is at most $\exp(-k/n) \leq \exp(-d/4)$. Let event $\cE_v$
occur when $v$ does not connect to $S$. By the union bound $\Pr[\bigcup_{v \notin S} \cE_v] \leq \sum_{v \notin S} \Pr[\cE_v]
\leq n \exp(-d/4) \leq n^{-4}$ (since $d \gg \ln n$). The event $\overline{\bigcup_{v \notin S} \cE_v}$ occurs precisely
when all vertices outside $S$ have an edge to $S$.
\end{proof}

\section{Modeling Path Length} 
\label{sec:level}

This section models the path lengths in the \sysname fabrics. We begin by bounding the sizes of Spraypoint levels.

\subsection{Level sizes} 
\label{sec:proof-level}

Let us fix a destination $t$, and for convenience,
denote $\wayp_{i}(t)$ as $\wayp_i$.
We imagine constructing
the levels incrementally, together with the graph. So $\wayp_{\leq i} \eqdef \bigcup_{j \leq i} \wayp_{j}$ is constructed, and then we choose
the randomness to construct $\wayp_{i+1}$. Conditioning on $\wayp_{\leq i}$
implies fixing the graph \emph{inside} $\wayp_{\leq i}$. So we have fixed all the edges between
the waypoints up to this level. Then $\wayp_i$ chooses some of its edges to determine $\wayp_{i+1}$.
It is convenient to set $\wayp_0$ to be $\nbr(t)$.

We state and prove our main theorem about level sizes. Recall that $\ell \eqdef \max(1,\lceil \log_p(n/2d^2) \rceil)$.

\begin{theorem} \label{thm:level} Fix destination $t$. The following bounds hold with probability $\geq 1-o(1)$.
We set $\lambda \eqdef p^\ell d^2/n$. 
\begin{asparaitem}
    \item For $0 \leq i \leq \ell$, $|\wayp_i(t)| = (1 \pm o(1)) dp^i$.
    \item $|\oring(t)| \leq (\exp(-\lambda) + o(1)) n$.
    \item All nodes in $\oring(t)$ have an edge to a node in $\inner(t)$.
\end{asparaitem}
\end{theorem}

The proof has a few moving parts, so we separate them out.
For any set of nodes $S$, we use $N(S)$ to denote
the neighborhood of $S$ (which may include nodes in $S$).

\begin{claim} \label{clm:connect} For any $i$, condition on $|\wayp_{\leq i}| \leq n/d$.
The probability that any node $v$ has more than $d/2$ edges to $\wayp_{\leq i}$
is at most $1/n$.
\end{claim}

\begin{proof} The node $v$ makes $d$ connections. Consider choosing these randomly as follows.
First, each half-edge incident to $v$ chooses to connect with a node in $\wayp_{\leq i}$
with probability $|\wayp_{\leq i}|/n \leq 1/d$. Once these choices are made, then 
the half-edges are paired to a random half-edge in $\wayp_{\leq i}$ or $\overline{\wayp_{\leq i}}$
respectively. Each of these choices can be represented as a Bernoulli $X_i$, and the sum
$X = \sum_{i \leq d} X_i$ is the number of edges $v$ has into $\wayp_{\leq i}$. 

We note that $\EX[X] \leq 1$. By the Chernoff bound \Thm{chernoff}, $\Pr[X \geq d/2] \leq 2^{-d/2}$.
By a union bound over all nodes, the probability that any node has more than $d/2$
edges is at most $n 2^{-d/2}$. Since $d \geq 2(\ln n + 5)$, $n2^{-d/2} \leq n^{-5}$.
\end{proof}

The following lemma encapsulates a key calculation.

\begin{lemma} \label{lem:rand-nbr} Condition on a set $S$ of nodes 
with at least half of its incident half-edges unpaired. Let $|S| \geq d/2$.
Consider pairing all the half-edges incident to $S$.
Then, with probability at least $1-n^{-10}$, $|N(S)| \geq \min(n/4, |S|d/4)$.
Moreover, if $|S| \leq n/d$, at least $(d-4)|N(S)|$ half-edges incident to $N(S)$ are left unpaired.
\end{lemma}

\begin{proof} We will perform a coupling of the random variable $|N(S)|$ as follows.
We initially mark an arbitrary set of $n/4$ nodes. 
We pair the half-edges incident to $S$ iteratively. If a pairing connects to an unmarked
node \emph{and} the number of current neighbors of $S$ is $< n/4$:
we pick an marked non-neighbor of $S$, unmark it, and marked the (new) unmarked neighbor.
(Otherwise, we do not change the marking.)
Thus, if there are fewer than $n/4$ neighbors of $S$, then all neighbors of $S$
are marked.

For the $r$th half-edge paired, let $X_r$ be the indicator of the $r$th half-edge
pairing with an unmarked node. Since the number of marked nodes is exactly $n/4$,
the number of unmarked nodes is also exactly $3n/4$. Hence, $\Pr[X_r = 1] = 3/4$.
Recall that there are at least $k \geq |S|d/2$ unpaired half-edges incident to $S$
Let $X = \sum_{r \leq k} X_r$, so $\EX[X] = 3k/4$.
By the Chernoff bound of \Thm{chernoff}, $\Pr[\sum_r X_r \leq k/2] \leq \exp(-\EX[X]/(3\cdot3^2))$
$\leq \exp(-|S|d/54) \leq \exp(-d^2/108)$. Since $d \geq \ln n$, this probability
can be bounded by $n^{-10}$ (for sufficiently large $n$).

So, with high probability, $X \geq k/2 \geq |S|(d/4)$. This means that either each iteration
connected to an unmarked node (thereby creating a new neighbor) or 
the number of neighbors was at least $n/4$. Thus, $|N(S)| \geq \min(n/4, |S|d/4)$.

If $|S| \leq n/d$, then $|S|/|N(S)| \leq 4/d$.
We paired up at most $|S|d$ half-edges in the above random process.
Hence, there are at least $|N(S)|d - |S|d = |N(S)|d - 4|N(S)|$
$\geq (d-4)|N(S)|d$ half-edges incident to $N(S)$ that are unpaired.
\end{proof}

\subsubsection{The first bullet point of \Thm{level}}
We prove the first bullet point by an induction over $i$. 
We require a stronger induction hypothesis. We will prove
that, conditioned on a choice of $\wayp_i$ that leaves at least $d|\wayp_i|/2$ 
half-edges unpaired, with high probability, the construction of $\wayp_{i+1}$
ensures that $|\wayp_{i+1}| \geq (1-o(1))pd^{i+1}$  and at least $d|\wayp_{i+1}|/2$
half-edges of $\wayp_{i+1}$ are unpaired.

We now perform the induction.
We first deal with the base case $\wayp_0$, the neighbor set. 
Consider the pairing the $d$ half-edges incident to $t$,
one by one. Take the $r$th half-edge being paired.
There are at most $d$ neighbors of $t$ already constructed, by the pairing of previous half-edges. The probability that the $r$th half-edge connects to one of these neighbors is at most $d/n$, regardless of all the previous choices. Let $Z$ be the random variable denoting the number of half-edges connecting to an existing neighbor. The upper tail probabilities of $Z$ is dominated by the corresponding tail of $B(d/n,d)$. (Here, $B(k,\alpha)$ denotes the binomial of $k$ independent trials with success probability $\alpha$.) We can apply the Chernoff bound to $B(d/n,d)$. Let $Y \sim B(d/n,d)$. Note that $\EX[Y] = (d/n)\cdot d = d^2/n$. Since $d \ll n$, for sufficiently small $\eps > 0$, $\eps d \geq 2ed^2/n$. By \Thm{chernoff}, 
$\Pr[Y \geq \eps d] < 2^{-\eps d} = o(1)$. 
Thus, $\Pr[Z \geq \eps d] \leq \Pr[Y \geq \eps d] = o(1)$.
In words, the overall probability of the half-edges connecting to $\eps d$ existing neighbors is $o(1)$. So, with probability $1-o(1)$, the size of the neighbor set $\wayp_0$ is at least $(1-o(1))d$.

At this stage, there are $(d-1)|\wayp_0|$ unpaired half-edges incident to $\wayp_0$.
The only paired half-edges (formed edges) are those incident to $t$, so each neighbor in $\wayp_0$
is incident to $(d-1)$ unpaired half-edges.

We split into two cases, depending on $\ell$.

{\textbf Case 1, $\ell = 1$:} In this case, there is just $\wayp_1$ to consider.
As argued above, with high probability,
$|\wayp_0| = d - o(1)$ and it is incident to $|\wayp_0|(d-1)$ unpaired edges. Apply \Lem{rand-nbr}
with $S = \wayp_0 = \nbr(t)$. So probability $> n^{-10}$, $|N(S)| \geq \min(n/4, |S|d/4)$.
Observe that $pd \leq n/4$ and $p \leq d/4$. So, $pd \leq |N(S)|$.
There are enough candidates to construct $\wayp_1$ as desired.
There is no induction necessary, since this is the last waypoint level.
This ends the proof for this case.

{\textbf Case 2, $\ell > 1$:} Note that $\ell = \lceil \log_p(n/2d^2) \rceil$. Pick some $i < \ell$ . We inductively assume that 
for all $j \leq i$, $|\wayp_j| = (1-o(1)) pd^j$ and that there are at least $|\wayp_i|d/2$
unpaired edges in the construction thus far. We just proved this for the base case.
At each step of the induction, there is a small probability of error ($< 1/n^5$). We will
simply union bound all these errors over the at most $\ell = \Theta(\ln n)$ levels.

Observe that $|\wayp_i| \leq dp^i \leq d p^{\log_p(n/2d^2)} \leq n/2d$. 
Also, $|\wayp_{\leq i}| \leq \sum_{j \leq i} dp^j = d(p^{i+1}-1)/(p-1)$.
Since $p \geq 2$, $p-1 \geq p/2$. Hence, $\sum_{j \leq i} |\wayp_i| \leq  d p^{i+1}/(p/2)$
$ \leq 2dp^i \leq (2+o(1)) |\wayp_i|$.

We apply \Lem{rand-nbr} with $S = \wayp_i$. Note that $|S| \geq d$ and $|S| \leq n/2d$.
At least half of the incident half-edges are unpaired. Hence, with probability $> 1-n^{-10}$,
$|N(\wayp_i)| \geq |\wayp_i|d/4$. Since $|\wayp_{\leq i}| \leq 2|\wayp_i|$,
there are at least $(d/4 - 2)|\wayp_i|\geq |\wayp_i|p$ (since $p+2 \leq d/4$) nodes that are neighbors of $\wayp_i$
and \emph{not} in $\wayp_{\leq i}$. These are all candidates for $\wayp_{i+1}$.
Thus, we can construct $\wayp_{i+1}$ as 
desired. We get that $|\wayp_{i+1}| = p|\wayp_i| = (1 \pm o(1)) dp^{i+1}$.

At least $(d-4)|N(\wayp_i)|$ half-edges incident to $N(\wayp_i)$ are unpaired.
It remains the bound the number of unpaired half-edges incident to $\wayp_{i+1}$.
Observe that a uniform random node of $N(\wayp_i)$ has at least $(d-4)$ unpaired
half-edges. Let $Z_r$ be the random variable denoting the fraction
of unpaired edges on the $r$th node of $\wayp_{i+1}$. We have $\EX[Z_r] \geq 1-d/4$.
Setting $Z = \sum_r Z_r$, $\EX[\sum_r Z_r] \geq (1-d/4)|\wayp_{i+1}|$ by linearity of expectation.
Applying \Thm{chernoff}, $\Pr[Z < \EX[Z]/2] \leq \exp(-(1/3\cdot 2^2) \cdot (1-d/4) |\wayp_{i+1}|)$.
Since $|\wayp_{i+1}| \geq (1-o(1)) pd^{i+1}$ and $i \geq 0$, the probability is $\exp(-\Theta(d)) \leq n^{-4}$.
Thus, we have at least $d|\wayp_{i+1}|/2$ unpaired half-edges incident to $\wayp_{i+1}$.
We union bound over all the errors.
This completes the induction, and the proof of the first bullet.

\subsubsection{The other bullets of \Thm{level}}

With high probability, $|\wayp_i| = (1 \pm o(1)) dp^i$. 
We will simply assume this property, and union bound the errors.
Hence, the set $\wayp_\ell$ is incident to at least $(1-o(1))d^2 p^\ell$ half-edges.
By \Clm{nbr}, the probability that a node is \emph{not} a neighbor of $\wayp_\ell$
is at most $\exp(-(1-o(1)) d^2p^\ell/n) = \exp(-\lambda(1-o(1))) = (1+o(1)) \exp(-\lambda)$.
For node $v$, let $X_v$ be the indicator random variable that $v$ is not a neighbor
of $\wayp_\ell$. Note that $\EX[X_v] \leq (1+o(1))\exp(-\lambda)$. 

Consider the graph construction process where each $v$ first independently chooses
to connect to $\wayp_\ell$. Then, if possible, the graph is formed conditioned on these choices. 
Setting $X = \sum_v X_v$, we can apply the Chernoff bound of \Thm{chernoff}. 
So $\Pr[|X - \EX[X]| > n/\ln n] \leq 2\exp(-2n/\ln^2n)$, which is a tiny probability.
Therefore, with high probability, the number of neighbors of $\wayp_\ell$ is 
close to the expectation, and hence, the graph can be feasibly formed conditioned
on the the choices of $X_v$.

Thus, with high probability, the graph construction process is feasible,
and the number of non-neighbors of $\wayp_\ell$ is at most $(1+o(1))\exp(-\lambda)n + n/\ln n
= \exp(-\lambda)n + o(n)$. Observe that all nodes in $\oring(t)$ are non-neighbors
of $\wayp_\ell$, so this completes the proof of the second bullet.

Now for the third bullet. We can bound 
\begin{equation*}
\lambda \geq p^{\lceil \log_p(n/2d^2) \rceil}d^2/n 
\geq 1/2
\end{equation*}
As shown earlier, with high probability, $|\oring(t)| \leq (\exp(-\lambda) + o(1))n
\leq 0.7n$. Hence, at least $0.3n \geq n/4$ nodes lie in the remaining levels, waypoints or $\inner(t)$.
\Clm{straggler} asserts that every node is a neighbor of these remaining levels with high probability.
So every node is either a neighbor of a waypoint or $\inner(t)$. Nodes in $\oring(t)$
are by definition not neighbors of waypoints, and hence must be neighbors of $\inner(t)$.
We union bound over all the errors, each of which was $< n^{-5}$. We perform the union
bound at most $n^2\log n$ times (once for each source, once for each destination, and
once for each level). That bounds the error probability as $n^{-3}\log n < n^{-2}$.

\subsection{Path length distribution} \label{sec:path}

We now compute the distribution of Spraypoint path lengths.

\begin{theorem} \label{thm:path-length} Fix a destination $t$. Consider the distribution of path lengths
seen by a packet sprayed uniformly at random from a uniform random source $s$.
With high probability,
the distribution of Spraypoint paths to a destination $t$ is as follows:
\begin{asparaitem}
    \item There are $d$ paths of length $1$.
    \item For every $2 \leq i \leq \ell + 2$, there are $(1 \pm o(1)) p^{i-2}d^2$ paths of length $i$.
    \item There are at most $(1 + o(1)) \exp(-\lambda) nd$ paths of length $\ell + 4$.
    \item All other paths have length $\ell + 3$.
\end{asparaitem}
\end{theorem}

\begin{proof} Consider the pointing structure. \Thm{level} gives the sizes
of each level in this structure. We will assume the bounds and statements of \Thm{level}.
Imagine routing traffic from every possible source to the destination $t$. The spraying step
would send traffic along all edges. For each edge, let us track the length to the destination.

Any edge incident to $t$ leads to a path of just length $1$, which gives $d$
paths of length $1$. If an edge incident to $\wayp_i(t)$, the path
to the destination has length $i+2$ (one step for spraying and $(i+1)$ steps along
the waypoints leading to $t$). Edges incident to $\inner(t)$ lead to paths of length
$\ell + 3$. Finally, by \Thm{level}, observe that all nodes in $\oring(t)$ have an edge to $\inner(t)$
but not to any waypoints. Hence, edges incident to $\oring(t)$ give Spraypoint paths
of length $\ell + 4$. We can use the level size bounds of \Thm{level} to complete the proof.
\end{proof}

The bounds in \Mod{path-length} are expressed as fractions. For a destination $t$, there are $nd$ possible paths to consider, one per source and spray operation. So we divide all the bounds in \Thm{path-length} by $nd$ to get the fractions in \Mod{path-length}.

\section{Modeling Edge disjoint paths} \label{sec:mincut}

In this section, we estimate the number of edge disjoint paths between a source $s$ and destination $t$ by estimating the min cut of Spraypoint paths between $s$ and $t$.
We first set up some notation.

\begin{definition} \label{def:parent} For every $v \neq t$, we define the $t$-parent as:
\begin{asparaitem}
    \item For $v \in \nbr(t)$, the parent is $t$.
    \item For $v \in \wayp_i(t)$ ($i \geq 1$): we pick $h$ uar neighbors in $\wayp_{i-1}(t)$ as the parents.
    \item For $v \in \iring(t)$: we pick up to $h$ uar neighbors of $v$ that are waypoints (excluding $\{t\} \cup \nbr(t)$).
    \item For $v \in \oring(t)$: we pick $h$ random neighbors in $\iring(t)$. 
\end{asparaitem}
(If $h$ selections cannot be made, we make the largest possible.)
\end{definition}

Recall that the routing algorithm is simple. The first step is Spraying. After that, as part of pointing, a packet is forwarded to a uar parent. Eventually, the packet reaches $t$. 
This leads to some natural definitions of Spraypoint paths, from which we can define the graph between the endpoints.

\begin{definition} \label{def:spgraph} Fix a destination $t$. For any $v$, the Spraypoint path to $t$
is any path obtained by repeatedly taking $t$-parents until ending at $t$.
\begin{asparaitem}
\item 
The \textbf{pointing graph} $\point{t}$ has all edges that forward traffic to $t$.
\item  For a source-destination pair $(s,t)$, the \textbf{Spraypoint graph} $\spg{s}{t}$ has all incident edges to $s$ and the union of Spraypoint paths from $N(s)$ to $t$.

\end{asparaitem}

\end{definition}

We work with $\ell = 1$ setting.
We make some asymptotic assumptions consistent with our parameters.
The assumptions allow for simpler formulas. We assume:
\begin{asparaitem}
    \item $pd/n \ll 1$
    \item $\lambda = pd^2/n \gg 1$
    \item $d \gg \ln^2n$
    \item $h \ll \lambda$
\end{asparaitem}

\begin{claim} \label{clm:pick-r1} Assume that $(p+1)d < n/\ln^2 n$. With probability $>1-n^{-2}$, all nodes pick at most $o(d)$
neighbors outside $\inner(t)$.
\end{claim}

\begin{proof} Fix a node $v$. We will treat each of the at most $d$ neighbors that $v$ has to pick
as independent choices among $\nbr(t) \cup \wayp_1(t) \cup R_1(t)$. The probability that a neighbor
lies in $\nbr(t) \cup \wayp_1(t)$ is at most $|\nbr(t) \cup \wayp_1(t)|/n \leq d(p+1)/n \leq 1/\ln^2n$. Let $X_v$
be the random variable denoting the number of neighbors of $v$ in $\nbr(t) \cup \wayp_1(t)$, 
which is a sum of iid Bernoullis. Note that $\EX[X_v] \leq d/\ln^2 n$.
By \Thm{chernoff}, $\Pr[X_v > d/\ln n] \leq 2^{-d/\ln n}$. By a union bound over all vertices,
the probability that any $X_v > d/\ln n$ is at most $n2^{-d/\ln n} \leq n^{-2}$.

By \Thm{path-length}, with probability $> 1-n^{-4}$, the size of $\oring(t)$ is at most $(\exp(-\lambda) + o(1))n$.
The probability that $v$ has a single neighbor in this set is $o(1)$. Using a Chernoff bound identical
to the one above, the probability that $v$ has more than $\delta n$ neighbors (for any constant $\delta > 0$)
is $n^{-\Theta(1)}$. 

Take a union bound, with high probability, all nodes pick at most $o(d)$ neighbors outside $\inner(t)$.
\end{proof}

\subsection{The matching connection} \label{sec:matching}

For any subset $S \subseteq \inner(t)$, let us construct the bipartite graph $\desc{S} = (S,\nbr(t),E)$
as follows. We connect $s \in S$ to $v \in \nbr(t)$ if $s$ is a descendant of $v$ in the pointing graph
$\point{t}$. 

\begin{claim} \label{clm:mincut} For any $S \subseteq \inner(t)$: connect a source $s$ to all nodes in $S$
and consider the resulting pointing graph.
The size of the $s$-$t$ mincut is the size of the largest matching in $\desc{S}$.
\end{claim}

\begin{proof} We first prove that the largest matching in $\desc{S}$ is at least the mincut size (say $r$).
By the maxflow-mincut theorem, there is a flow of value $r$ from $S$ to $t$, with unit edge capacities.
By the integrality of flow, this flow consists of $r$ edge disjoint paths from $s$ to $t$. Note that each node
in $S$ has a single edge to $s$, and every node in $\nbr(t)$ has a single edge to $t$. So 
a node in $S \cup \nbr(t)$ participates in at most one flow path. Thus, the $r$ flow paths provide
edge disjoint paths from $S$ to $\nbr(t)$ in the pointing graph $\point{t}$, which constitute a matching in $\desc{S}$.

Now, we prove that the mincut is at least the size of the largest matching in $\desc{S}$. Consider a matching
given by $\phi:S \to \nbr(t)$. For each $s' \in S$, follow a parent to $\wayp_1(t)$, and then a parent to get to $\phi(s) \in \nbr(t)$.
For two nodes $s' \neq s''$ in $S$, $\phi(s') \neq \phi(s'')$. So the corresponding paths are edge disjoint.
Hence, we can send one unit of flow on each such path, and get a flow with value at least the matching size.
The maxflow value (and hence the mincut value) is at least the matching size.
\end{proof}

We can now connect the $s$-$t$ mincut of the Spraypoint graph $\spg{s}{t}$ to matching sizes. Note that $\spg{s}{t}$
is formed by adding the edges connecting $s$ to its neighborhood to $\point{t}$. For convenience,
we will only add the edges connecting $s$ to $\inner(t)$. By \Clm{pick-r1}, this reduces
the mincut by at most $o(d)$.

To bound the matching sizes, we use the following notation, based on random matching sizes.

\begin{definition} \label{def:random-match} For positive integers $\ell, r, h$, consider 
the random bipartite graph $H(L,R,E)$ formed as follows. We have $|L| = \ell$, $|R| = r$, and
each node in $L$ picks $h$ uar neighbors in $R$ with replacement. We use $\mu(\ell,r;h)$ to denote
the expected maximum matching in $H$.
\end{definition}

We now state our main theorem relating $s$-$t$ mincuts in $\spg{s}{t}$ to matching sizes.

\begin{theorem} \label{thm:mincut} Assume $h = o(d)$. The expected $s$-$t$ mincut size of $\spg{s}{t}$
has the following values:
\begin{asparaitem}
    \item If $s \in \nbr(t)$: the size is at least $\mu(d-p,d;h) \pm o(d)$.
    \item If $s \notin \nbr(t)$: the size is at least $\mu(d,d;h) \pm o(d)$.
\end{asparaitem}
\end{theorem}

\begin{proof} Consider $s \in \nbr(t)$. As discussed earlier, it picks $d-p-1$ random neighbors
after conditioning on the pointing structure. By \Clm{pick-r1}, at most $o(d)$ of these neighbors
lie outside $\inner(t)$. So, it picks $d-p-1-o(d) = d-p-o(d)$ neighbors in $\inner(t)$. Call this set
$S$. By \Clm{mincut}, the $s$-$t$ mincut size is the size of the largest matching in $\desc{S}$.
Each node $s'$ in $S$ (which is in $\inner(t)$) picks $h$ random parents in $\wayp_1(t)$, each of which
has a parent in $\nbr(t)$. (We ignore the fact that nodes in $\wayp_1(t)$ might have multiple parents.)
Thus, each $s'$ effectively picks $h$ random ancestors in $\nbr(t)$.
So $\desc{S}$ is distributed exactly as the random bipartite graph in \Def{random-match},
and the expected matching size is $\mu(d-p-o(d),d;h)$. We can express this quantity
as $\mu(d-p,d;h) \pm o(d)$, since changing any of the sets (in a  bipartite graph)
by $o(d)$ can affect the matching size by at most $o(d)$.

Consider $s \in \wayp_1(t)$. It picks $d-o(d)$ random neighbors in $\inner(t)$. Applying the same
logic as above, the expected mincut has value $\mu(d,d;h) \pm o(1)$.
For $s \in \inner(t)$, it picks $d-h-o(d)$ random neighbors in $\inner(t)$. Since $h=o(d)$,
$s$ basically picks $d-o(d)$ random neighbors, and the value is the same as for $s \in \wayp_1(t)$.
\end{proof}

This theorem gives a lower bound on the size of the mincut. As mentioned in the proof,
we ignore the fact that $\wayp_1(t)$ has potentially $h$ neighbors into $\nbr(t)$, and only
use one of the neighbors for routing. We believe this only loses
a lower order term, and so we ignore it for the sake of cleaner expressions.

\subsection{Expressions for random matching sizes} \label{sec:rand-match}

A good rule of thumb is that $\mu(d,d;h) \approx d(1-\exp(-h))$. If $d-p = \alpha d$,
then: 
\begin{equation}
    \mu(d-p,d;h) = \mu(\alpha d, d;h) = \min\Big[\alpha d, d[1-\exp(-\alpha h)]\Big]
\end{equation}

For $h=1$, these bounds are tight. For larger $h$, there
are only upper bounds. 
Nevertheless, they match up quite accurately with experiments.
The exact lower bound is a more complex expression given by Frieze-Mellsted~\cite{FrMe12},
as explained later.

We show the \emph{upper} bound in the next lemma, and then give the complicated exact expression. (The proof for $\mu(\alpha d, d; h)$ is analogous and omitted.)

\begin{lemma} \label{lem:match} For all $h$, $\mu(d,d;h) \leq d(1-\exp(-h)) - o(d)$.
Also, $\mu(d,d;1) = d(1-1/e) - o(d)$.
\end{lemma}

\begin{proof} Recall that $\mu(d,d;h)$ is the expected maximum matching in the following
random bipartite graph. There are two sets $L,R$ of size $d$. Each node in $L$ makes $h$
random connections with $R$.

Consider a vertex $r \in R$. The probability that an edge does not connect with $r$
is $(1-1/d)$. The probability that no edge makes a connection is $(1-1/d)^{dh}$,
since there are $dh$ random (multi)edges in the graph. Since $1-x \geq \exp(-x - x^2)$,
we lower bound this probability by $\exp(-(1/d + 1/d^2)\cdot dh) = \exp(-h)(1-\Theta(h/d))$.

Consider indicator random variable $X_r$ for $r$ having degree zero.
By linearity of expectation and the calculation above, $\EX[\sum_{r \in R} X_r] = \sum_{r \in R} \EX[X_r]
\geq d \exp(-h) - o(d)$. We can also upper bound $(1-1/d) \leq exp(-1/d)$,
and get the lower bound $d\exp(-h)$. Thus, the expected number of degree zero
nodes in $R$ is $d\exp(-h) \pm o(d)$. These nodes cannot be matched, and hence the
matching size is at most $d(1-\exp(-h)) \pm o(d)$.

When $h=1$, we can match this bound. 
Take the iterative process that picks an arbitrary node in $R$ with non-zero degree, and matches
it arbitrarily to some neighbor in $L$. Then, we delete this node and its neighbors from the graph,
and iterate. When $r \in R$ is removed, it removes its neighbors in $L$. Since they all have degree $1$,
all edges removed are incident to $r$. Hence, degrees of other nodes in $R$ are unaffected. All in all,
every node of non-zero degree in $R$ gets matched.
\end{proof}

\textbf{The Frieze-Mellsted bound:} The optimal value of $\lim_{d \to \infty}$ $\mu(d,d;h)/d$
was computed in~\cite{FrMe12}. We state the bound here, and show that it
is fairly close to the simpler expression of \Lem{match}. 

(Refer to Thm. 3 of~\cite{FrMe12}.) Let $z^*$ be the largest non-negative solution of $(z/h)^{1/(h-1)} + \exp(-z) + 1 = 0$.
Then,
\begin{eqnarray}
    \lim_{d \to \infty} \mu(d,d;h)/d & = & 2 - (1-\exp(-z^*))^h \nonumber \\
    & & - (1+z^*)\exp(-z^*) \\
    & \approx & 1 - (1+z^* - h)\exp(-z^*)
\end{eqnarray}
For $h=2$, the bound is $\approx 0.838$, while the bound of \Lem{match} is $\approx 0.865$.
For $h=4$, the bound is $\approx 0.979$, while the bound of \Lem{match} is $\approx 0.982$.
As $h$ becomes larger, $z^*$ tends to $h$, and the bound above tends to $1-\exp(-h)$.

\subsection{The mincut bound, summarized} \label{sec:mincut-summary}

We summarize the bounds in this section. Fix a destination $t$.
The $s$-$t$ mincut value of the Spraypoint graph $\spg{s}{t}$ is given by the following 
formulas:

\begin{itemize}
    \item For most $s$, the average mincut is $d(1-\exp(-h))$.
    \item The lowest (expected) mincut occurs for $s \in \nbr(t)$. In that case, setting $d-p = \alpha d$, the bound is
    $\min\Big[\alpha d, d[1-\exp(-\alpha h)]\Big]$. When $\alpha$ is close to $1$, the bound is close to $d(1-\exp(-h))$.
\end{itemize}
To get a sense of the benefits of $h$, we plug in different values. The best possible
mincut is $d$. Let us see what fraction of $d$ is achieved.
For $h=1$, we get a $1-1/e \approx 0.63$ fraction. For $h=2$, we get a $1-\exp(-2) \approx 0.86$ fraction.
For $h=4$, we get a $1-\exp(-4) \approx 0.98$ fraction. At this point, the statistical variation of the mincut
value is bigger than any potential improvements. 

The following table gives the various fractions. Note that $p=0$ gives the bound for most $s$.

\begin{center}
    \begin{tabular}{l || l | l | l |}
    $p$ & $h=1$ & $h=2$ & $h=4$ \\ \hline
    $0$ & $0.63$ & $0.86$ & $0.98$ \\ \hline
    $d/4$ & $0.53$ & $0.75$ & $0.75$ \\ \hline
    $d/3$ & $0.49$ & $0.66$ & $0.66$ \\ \hline
    $d/2$ & $0.39$ & $0.5$ & $0.5$ \\ \hline
    \end{tabular}
\end{center}

\section{Modeling Oversubscription} 
\label{sec:app:oversub}

To help us characterize stochastic oversubscription, 
we set up some notation and definitions. We set up a multicommodity flow
problem. Each edge of the graph $G = (V,E)$ is treated as two directed arcs,
one in each direction. The capacity of each arc is one. There is a doubly
stochastic demand matrix $M$ that specifies that amount of flow
that each source needs to send to each destination. So $M_{s,t}$
is the amount of flow that $s$ needs to send to $t$. The total flow in or out of any node is at most one since $M$ is doubly stochastic.

\begin{definition} 
\label{def:oversub} (1) A demand matrix $N$ is a \textbf{feasible matrix} if the multicommodity flow specified by $N$ can be routed while honoring link capacity constraints. (2) The \textbf{demand multiplier} of a doubly stochastic demand matrix $M$,
denoted $c(M)$, is the largest $c$ such that $cM$ can be satisfied. (3) The \textbf{oversubscription ratio (oversub)} is the value of $\max_{M} d/c(M)$.
\end{definition}

Thus, the oversub is the scaled reciprocal of the lowest possible
demand multiplier.

Our first lemma connects the oversub to the fraction of optimal flow along different lengths.

\begin{lemma} \label{lem:oversub} Consider the traffic matrix $M$ that achieves
the oversub bound. For the corresponding optimal flow, let $\delta_i$
be the fraction of the flow that takes length $i$ paths. 
Then, the oversub is at least $\sum_{i} i \delta_i$. 
\end{lemma}

\begin{proof} The proof goes via a simple ``total capacity" argument.
Suppose the oversub is $\omega$. This means that the corresponding
demand multiplier is $d/\omega$, which is the total flow that each node
sends/receives. The total flow in the network is $nd/\omega$, since
each node is a source. 

Let us bound the sum $F \eqdef \sum_{\textrm{arc} \ e} f_e$, where $f_e$ is the total flow on the arc $e$. 
By definition, the amount of flow along length $i$ paths is $\delta_i nd/\omega$. 
Observe that this flow takes up $i \delta_i nd/\omega$ capacity. Thus,
the contribution of this flow to the sum $F$ is $i \delta_i nd/\omega$,
and we bound $F = \sum_i i \delta_i nd/\omega$.

Since each arc has unit capacity and there are $nd$ arcs, $F \leq nd$.
So $\sum_i \delta_i nd/\omega \leq nd$ and $\omega \geq i \sum_i \delta_i$.
\end{proof}

\subsection{Principles underlying the model} 
\label{sec:oversub-principles}

Based on \Lem{oversub}, we will devise a model for oversub in the next section as a function of system parameters, $n, d, p, h$.  The model is based on a number of principles.

{\bf The greedy shorter path principle:} The oversub lower bound of \Lem{oversub} is $\sum_i i \delta_i$,
so it is preferable to have larger values of $\delta_i$ for small $i$, rather than the other way around. 
Based on this intuition, we construct
the demand multipliers for a demand matrix by greedily maximizing the amount
of flow on shorter paths first. For the $\ell = 1$ setting, all flows
paths have length between $1$ to $5$ (\Thm{path-length}). We can ignore path length $1$,
which form a neglible fraction of paths. The main calculations compute the maximum
amount of flow that goes along paths of length $2$, then conditioned on that, maximize
flow along paths of length $3$, so on and so forth. We compute $\mu_2, \mu_3, \mu_4$,
and $\mu_5$ defined as follows. The average amount of flow that a source sends
along length $i$ paths is $\mu_i d$. The total demand multiplier is $(\mu_2 + \mu_3 + \mu_4 + \mu_5)d$,
and hence the oversub is $(\mu_2 + \mu_3 + \mu_4 + \mu_5)^{-1}$.

{\bf The average case principle:} By Theorem 2.1 of~\cite{throughput-centric2021}, the worst-case traffic 
matrix is a permutation matrix, which is a perfect matching of source to sinks. We will actually analyze the average demand multiplier, and hence oversub, of a \textit{random} permutation matrix. In many problems in random graph theory, the worst-case behavior can be shown to be close to the average-case behavior. Specifically, suppose an adversary picks a traffic matrix and then we choose the graph $G$. We can compute the probability that the demand multiplier deviates significantly from the average value. If this probability is small enough (like $< n^{-n}$), we could union bound over all permutation matrices. This means, that with non-trivial probability, a random graph will have high demand multipliers for \emph{all} permutation matrices. This is a common paradigm in random graph/matrix theory: show that the probability of large deviations from average are exponentially small, and then union bound over all the exponential possibilities (proof of expansion of random graphs, Chap. 5 of~\cite{MoRa-book}). 

{\bf The random deletion principle:} Once (say) we have determined $\mu_2$,
this flow uses up some capacity. When computing $\mu_3$,
the capacity constraints have changed. We consider a simplified
model for the change. The length $2$ flow uses up $2\mu_2 nd$ capacity
from the network, so we assume that each edge is fully congested
with probability $2\mu_2nd/nd = 2\mu_2$. We delete these edges, since they are unavailable to carry flow. This 
deletion principle is less accurate for larger length flows. Larger length flows do not use up the capacity of constituent edges. Hence, for most calculations,
we will only assume this principle for smaller length flows.

{\bf The viable path principle:}  For any length $i$, we estimate the number of paths of length $i$ that can be used
for routing. For concreteness, for the remaining principles, we will consider length $3$ (the first non-trivial case).
Such a path goes $s \rightarrow \wayp_1(t) \rightarrow \nbr(t) \rightarrow t$.
\Thm{path-length} tells us how many neighbors of $s$ lead to such paths
(in this case, it is $pd^2/n$ for an average source $s$). 
Because of the deletion principle, some of these edges (from $s$ to $\wayp_1(t)$)
get deleted. Consider some $v \in \wayp_1(t)$ such that $(s,v)$ is not deleted.
The pointing paths from $v$ to $t$ may get removed \emph{after} the deletions.
We need to model the probability of this event. The parameter $h$
ensures that each $v$ has $h$ (edge disjoint) pointing paths to $t$.
We calculate the probability that some path survives the deletion.
Putting it all together, we can estimate the number of paths of length $3$
that allow for routing. 
The next principles account for the intersections/congestion among these paths,
from which we compute the flow that can be sent along these paths.

{\bf The unique position principle:} 
By the mincut bounds of \Lem{match}, there should be enough edge disjoint
paths length between a given source-destination pair to carry the desired flow.
(We do not expect an oversub less than $2$, so even an mincut of $d/2$
suffices for a source-destination pair.) 
Consider the entire collection of edge disjoint length $3$ paths, over all source-destination pairs.
Congestion occurs because some edges
are present among length $3$ paths for multiple source-destination pair.

The concept of \emph{position} helps us bound the congestion. For a path,
an edge has one of three positions: first, middle, or last. \emph{An edge $(u,v)$
cannot be in the first position for more than one path.} When $(u,v)$
is the first edge, $u$ must be the source, and all length $3$ paths
from $u$ (in the collection) are edge disjoint. Similarly, an edge
cannot be in the last position for more than one path. Extrapolating,
we assume that no edge can be in a specific position for more than one path.
(This is technically false, since an edge could be in the middle for multiple
paths, but we believe this only leads to a lower order term.) Overall,
for length $3$, an edge can participate in at most $3$ flow paths of length $3$. Similarly, we assume that for length $i$,
an edge participates in at most $i$ flow paths, one in each position.

{\bf The binomial congestion principle:} We want
to compute the amount of flow along length $3$ paths to estimate $\mu_3$.
Using the viable path principle, suppose we determine that there are $\phi_3 d$ edge disjoint
paths of length $3$ between every source-destination pair. Consider the total
collection of $\phi_3 nd$ paths. We want to model, for a given edge, the number 
of paths that this edge participates in. This is the congestion on that edge,
focusing only on length $3$ flow. By the unique position principle,
the congestion is at most $3$. We will assume that an edge is present in a particular
position independently with probability $\phi_3$. Consider an edge on a particular
length $3$ path. It already occupies a position on this path, and can take
up two other positions (on different paths). The congestion on the edge
is distributed as the binomial $B(2,\phi_3) + 1$ (the binomial  $B(2,\phi_3)$ is the sum of $2$ independent $0$-$1$
random variables with expectation $\phi_3$) plus one. 

{\bf The reciprocal principle:} 
On each edge of the path, using the binomial model, we have a distribution
for the number of \emph{other} paths that use this edge. The maximum value among
all the edges of the path is the congestion of the path.
We use the following simple claim to get a flow bound.

\begin{claim} \label{clm:flow-cong} Consider a family $\cP$ of paths. For
each path in $P \in \cP$, let $c(P)$ be the congestion of $P$ among the paths
$\cP$. (So $c(P)$ is the maximum, among all edges of $P$, of the number
of paths that use the edge.) Then each path $P \in \cP$ can simultaneously route $1/c(P)$ units
of flow without violating the (unit) capacity constraints.
\end{claim}

\begin{proof} Suppose every path $P \in \cP$ routes $1/c(P)$ units of flow.
Consider any edge $e$ that is contained in path $P_1, P_2, \ldots, P_k$.
The congestion of $e$, denoted $c(e)$, is $k$. Observe that for every path $P_i$,
$c(P_i) \geq c(e) = k$. The total amount of flow in $e$
is $\sum_{i \leq k} 1/c(P_i) $ $\leq \sum_{i \leq k} 1/c(e) = \sum_{i \leq k} 1/k = 1$.
\end{proof}

This claim allows us to get the total flow that can be sent along, say, length $3$
paths. Over these paths, we have computed the distribution of congestion.  The average value of the reciprocal of the congestion is the average flow that can be sent along paths of this length. We multiply this average by the total number
of paths to get the total flow.

To move to the next path length, we use the random deletion principle to delete
edges that are used by this flow. We run this entire calculation again for the next path length.
For larger paths lengths $i$ (like $4$ and above), the congestion values of the paths 
are quite close to their maximum value $i$. So we just send $1/i$ units of flow along paths of these
length. This simplification helps us avoid complicated probability calculations that are 
lower order terms.

\subsection{The oversub formula} 
\label{sec:oversub-model}
\label{sec:formula}

We now compute the formula based on the principles above. Given the parameters, $n$, $d$, $p$, and $h$, the final formula is a polynomial function of $p$ and $d/n$, with exponential
dependencies on $h$. There are minor corrections needed for extreme parameter
settings (when $d$ is comparable to $n$, or $p$ is close to $d$).
As discussed earlier, we will estimate $\mu_2, \mu_3, \mu_4$, and $\mu_5$.

\underline{$\mu_2$:} By \Thm{path-length}, for a given destination $t$,
a typical source $s$ would have a $d/n$ fraction of its neighbors in $\nbr(t)$. 
The total number of such paths from $s$ to $t$ is $d^2/n$ and the total number
of edges on these paths is $2d^2/n$. As a fraction of the total number of arcs (directed edges),
we get $2d/n$ ($= (2d^2/n)\cdot n/nd$). We typically assume that $2d \ll n$.
Assuming that an edge is present on these paths with probability $2d/n$,
the probability that an edge is on two (or more) paths is $(2d/n)^2$, which is
negligible. So we assume that no edge is present on two of these paths,
or equivalently, the congestion of every path/edge is at most one.
So all the length $2$ paths across all source-destination pairs are edge disjoint.
Each such path can carry one unit of flow, so $s$ can send $d^2/n$ units
of flow along length $2$ paths to $t$. Thus, $\mu_2 = d/n$

\underline{$\mu_3$:} This is where most of the work goes. By the unique position principle
and the random deletion principle, each edge has probability $\mu_2$ of being in the first
position of a length $2$ flow path. (Same for the last position.) So we will delete
a $\mu_2$ fraction of edges (for the first position), and delete another $\mu_2$ fraction
for the last position. 

We now estimate the number of length $3$ paths that can be used to route flow.
We compute a fraction $\phi_3$ such that each source-destination pair has $\phi_3 d$
viable length $3$ paths to route flow. 
By \Thm{path-length}, a typical source $s$ has a $pd/n = p\mu_2$ fraction of neighbors in $\wayp_1(t)$.
(For extreme large values of $p$, $p\mu_2$ could be larger than one, so we ``correct"
by making this fraction $\min(p\mu_2, 1-\mu_2)$. The latter term just assumes that everything that
is not in $\nbr(t)$ is in $\wayp_1(t)$.) 
Of these corresponding edges incident to $s$, some of them might be in the last position of a length $2$ flow.
They cannot be in the first position, since all first position edges go from $s$ to $\nbr(t)$,
while these edges under consideration go from $s$ to $\wayp_1(t)$. By the deletion principle, we expect a $\mu_2$-fraction
of the edges to be deleted.

So we get that a $\min(p\mu_2, 1-\mu_2) \cdot (1-\mu_2)$ fraction of neighbors of a source $s$
lead to length $3$ paths to the corresponding destination $t$. Consider such a neighbor $v$ of $s$
that is in $\wayp_1(t)$. It has $h$ pointing (likely edge disjoint) paths of length $2$ to the destination $t$. 
Each edge gets deleted with independent probability $2\mu_2$, so a path survives
with probability $(1-2\mu_2)^2 \approx 1-4\mu_2$ (assuming $\mu_2 << 1$). 
The path gets removed with probability $\approx 4\mu_2$.
The probability that at least one of $h$ edge disjoint paths survives is $1 - (4\mu_2)^h$. Overall, the probability that $v \in \wayp_1(t)$ still
has a path to $t$ is $1 - (4\mu_2)^h$. 

Combining with the number of such $v$'s, we get that there are $\phi_3 d$ neighbors
of $s$ that lead to viable $3$ length flows, where:
\begin{equation}
    \phi_3 = \min(p\mu_2, 1-\mu_2) \cdot (1-\mu_2) \cdot (1 - (4\mu_2)^h)
\end{equation}
Our next step is to deal with the congestion. Consider a length $3$ flow path.
Each edge on this path could potentially take up two other positions (on different paths),
using the unique position principle. The congestion on the path
is the maximum of three independently chosen random variables according to the shifted binomial $B(2,\phi_3) + 1$.

\begin{claim} \label{clm:max-bin} Define $Y \eqdef \max(X_1, X_2, X_3)$, where each $X_i$ is independently chosen from $B(2,\phi_3)$.
Then, $\Pr[Y = 0] = (1-\phi_3)^6$, $\Pr[Y = 1] = [(1-\phi^2_3)^3 - (1-\phi_3)^6]$,
and $\Pr[Y = 2] = [1-(1-\phi^2_3)^3]$.
\end{claim}

\begin{proof} We can treat each binomial as the sum of two independent Bernoullis, each
with parameter $\phi_3$. For $Y$ to be zero, all the Bernoullis need to take value zero.
This happens with probability $(1-\phi_3)^6$.

Each $X_i$ takes value $2$ with probability $\phi^2_3$. The probability that no $X_i$
takes value $2$ is $(1-\phi^2_3)^3$. Thus, $\Pr[Y = 2] = [1-(1-\phi^2_3)^3]$.

We deduce $\Pr[Y = 1]$ as $1 - \Pr[Y=0] - \Pr[Y=2]$, and applying the above formulas.
\end{proof}

The congestion on a path is distributed as $Y+1$, where $Y$ is distributed
as in \Clm{max-bin}. As discussed in the earlier section, the average
flow that can be sent is $\EX[1/(Y+1)]$. Using \Clm{max-bin},
this average can be computed easily:
\begin{eqnarray}
    \kappa_3 & = & (1-\phi_3)^6 + [(1-\phi^2_3)^3 - (1-\phi_3)^6]/2 \nonumber \\
    & & + [1-(1-\phi^2_3)^3]/3 
\end{eqnarray}
Finally, we set $\mu_3 = \phi_3 \kappa_3$. 

\underline{$\mu_4$ and $\mu_5$:} We create variables
for the relative sizes of the Spraypoint levels, based on \Thm{path-length}.
\begin{equation}
    \sigma_4 = 1 - (p+1)d/n - \exp(-pd^2/n) \ \ \ \ \ \sigma_5 = \exp(-pd^2/n)
\end{equation}
Let us start with $\mu_4$. The quantity $\sigma_4$ is the fraction of $s$'s
neighbors that lead to a length $4$ flow path. The probability
that these corresponding edges are removed by previous flows is $\mu_2 + 2\mu_3$.
(This edge can take up the last position in length $2$ flows, and
the middle or last position in length $3$ flows.)

Consider a neighbor $v$ of $s$ leading to length $4$ path. As in the previous
$\mu_3$ calculations, we need to compute the probability that $v$ still has
a path to $t$ after the deletions. We will model $v$'s pointing paths as follows.
It has $h$ neighbors in $\wayp_1(t)$, each of which has $h$ edge disjoint paths
to $t$. Since we are considering longer paths, we will go with the weaker deletion
principle of removing edges with probability $2\mu_2$. The probability
that $v$ still has a path to $t$ can be approximated as $\beta$,
defined below.
\begin{equation}
    \beta = 1 - [1- (1-2\mu_2)(1-(4\mu_2)^h)]^h
\end{equation}
The fraction of viable paths is $\phi_4 = \sigma_4 \beta$.
For paths of length $4$, we will simply assume the worst case congestion of $4$,
so we can only send $1/4$ units of flow on each path. So we set $\kappa_4 = 1/4$ and
 $\mu_4 = \sigma_4 \kappa_4 $ $=\sigma_4 \beta (1-\mu_2 - 2\mu_3)/4$.

For $\mu_5$, we only look at deletions incident to the source $s$.
This is because there will be many paths (polynomial in $h$)
from a neighbor of $s$ (in $\oring(t)$) to the destination $t$. The odds
are quite high that some path will still be available for routing, so we do not bother
to model it carefully. 
Following the logic for $\mu_4$, we estimate
$\mu_5 = \sigma_4 (1-\mu_2 - 2\mu_3 - 3\mu_4)/5$.

Let $\mu^+_i = \max(\mu_i, 0)$. This deals with extreme situations ($n$ is too small,
$p$ or $d$ are large).
Finally, our oversub bound is
\begin{equation}
(\mu^+_2 + \mu^+_3 + \mu^+_4 + \mu^+_5)^{-1}
\end{equation}

\end{appendix}

\end{document}